\documentclass[12pt]{article}
\usepackage{tikz}
\usetikzlibrary{arrows.meta,positioning,calc}
\tikzset{>=Stealth}
\usepackage{amsfonts}
\usepackage{fancyhdr}
\usepackage{comment}
\usepackage{natbib}
\usepackage{caption}
\usepackage[a4paper, top=2.5cm, bottom=2.5cm, left=2.2cm, right=2.2cm]{geometry}
\usepackage{times}
\usepackage{amsmath}
\usepackage{multirow}
\usepackage{changepage}
\usepackage{colortbl}
\usepackage{booktabs}
\usepackage{amssymb}
\usepackage{graphics}
\usepackage{tikzsymbols}
\usepackage[separate]{patty}
\usepackage{float}

\newcounter{example}
\newcommand{\exampleNum}[1]{\refstepcounter{example}\label{#1}}

\begin{document}

\title{Classification in Equilibrium: Structure of Optimal Decision Rules\footnote{ We thank Bailey Flanigan, Moritz Hardt, Ben Klemens, Juan Carlos Perdomo, Chara Podimata, and Alexander Tolbert for very helpful conversations and feedback.}}

\author{Elizabeth Maggie Penn\thanks{Professor of Political Science, Professor of Data \& Decision Sciences, Emory. \newline Email: \textit{elizabeth.m.penn@gmail.com}.}\; and John W. Patty\thanks{Professor of Political Science, Professor of Data \& Decision Sciences, Emory. Email: \textit{jwpatty@gmail.com}. }}

\date{\today}

\maketitle

 \begin{abstract}
 
 This paper characterizes optimal classification when individuals adjust their behavior in response to the classification rule.  We model the interaction between a designer and a population as a Stackelberg game: the designer selects a classification rule anticipating how individuals will comply, cheat, or abstain in order to obtain a favorable classification. Under standard monotone likelihood ratio assumptions, and for a general set of classification objectives, optimal rules belong to a small and interpretable family---single-threshold and two-cut rules---that encompass both conventional and counterintuitive designs. Our results depart sharply from prior findings that optimal classifiers reward higher signals. In equilibrium, global accuracy can be maximized by rewarding those with lower likelihood ratios or by concentrating rewards or penalties in a middle band to improve informational quality.  We further characterize classification objectives that rule out socially harmful equilibria that disincentivize compliance for some populations.

 \end{abstract}

\normalsize

\section{Introduction}

A central challenge in designing decision rules that classify people is that the rule itself reshapes the behavior it evaluates. Once individuals understand how they are being assessed, they adjust their actions in order to obtain better outcomes, so that classification alters the  distribution of behavior it is meant to predict. This endogeneity lies at the core of a rapidly expanding literature in incentive-aware and strategic machine learning that analyzes how individuals manipulate observable features to influence algorithmic decisions \citep{HardtEtAl16,Dong18,HuEtAl19,MilliEtAl19,Podimata25},  and a complementary literature on performative prediction that examines how deployed models change the data-generating process that they are trained on \citep{PerdomoZrnicMendlerDunnerHardt2020,KimPerdomo23,HardtMendler-Dunner23}.  At the same time, a longstanding tradition in economics and political science examines how evaluative rules shape individual behavior. Signaling and screening theories \citep{Spence73,RothschildStiglitz76,Holmstrom82}, statistical discrimination \citep{Phelps72,Arrow73Discrimination,FangMoro11}, and work on affirmative action and institutional design \citep{McCubbinsSchwartz84,CoateLoury93} all highlight how thresholds, monitoring, and classification criteria can induce strategic responses in a population.

This existing research typically characterizes how individuals respond to a fixed classifier or policy rule---whether through genuine investment, effort, or feature manipulation---but rarely characterizes the designer's optimal rule under general objectives that explicitly anticipate those responses. In contrast, we characterize the globally optimal classifier under general designer payoffs and fully endogenous behavioral responses.  In our setting a designer selects a classification rule, agents choose whether to comply (or, in an extension, to possibly ``cheat''), and the resulting equilibrium prevalence of compliance depends on the incentive properties of the classifier itself. We show that, under a monotone likelihood ratio assumption, we can completely characterize globally optimal classification rules across a wide range of classification objectives and constraints.

Our first main result (Theorem \ref{quotaTheorem}) shows that in a baseline environment,  every optimal rule is a single threshold rule that may be either positive (rewarding scores falling above a ``cut") or negative (rewarding scores \textit{below} the cut). The possibility of an optimal negative threshold rule is surprising: we find that even accuracy-maximizing designers may optimally reward low-likelihood-ratio signals in order to reduce equilibrium compliance and thereby improve global accuracy. This creates the potential for endogenous behavioral disparities to emerge in equilibrium, where one population faces a positive threshold that increases compliance, while another faces a negative threshold that decreases it.

We then identify conditions on the designer's classification objectives under which this pathology cannot occur (Theorem \ref{noNeg}). A set of simple order restrictions on the designer's payoffs rules out negative thresholds entirely, ensuring that optimal classifiers never strictly disincentivize compliance in a population. These conditions draw a sharp contrast between accuracy objectives---which can generate opposing incentive effects even when the same classification objective is applied across groups---and alternative classification objectives such imposing a lower bound on precision.

Finally, we extend the model to environments where individuals may choose an intermediate ``cheating" action that partially mimics compliance at lower cost. Here we prove that optimal rules remain interpretable for a general class of signal environments: every optimal classifier is either a single threshold or a ``two-cut" rule, which rewards only a middle band of scores or its complement (Theorem \ref{twoCuts}). Cheating consequently breaks the monotonicity properties of an optimal classifier. In some cases, a two-cut rule can reduce mimicry and produce a cleaner, more separable empirical distribution of signals.

Taken together, our results show that optimal classification remains low-dimensional and interpretable even when behavior endogenously responds to the rule. However, the optimal design of a classifier can depart sharply from the monotone structures familiar in machine learning and standard mechanism design environments. Our results also illustrate that designer objectives matter a lot: natural adjustments to the designer's payoff function can eliminate incentive pathologies entirely. The remainder of the paper formalizes these results and situates them within the broader ML and social science literatures on strategic behavior, screening, and performative feedback.

\paragraph{Our contribution.} We provide a full characterization of globally optimal classification for general classification objectives, under potential capacity constraints, and when individual behavior responds endogenously to the classification rule.  

\begin{enumerate}
\item \textit{Threshold optimality}: In our baseline environment, and under possible capacity or quota constraints on the allowable fraction of positive classifications, every optimal classification rule is positively or negatively monotone and characterized by a single threshold.
\item \textit{Constraints on  objectives}: We identify simple inequalities on a designer's classification objectives that rule out interior negative threshold optimality and therefore prevent optimal classifiers from strictly incentivizing socially harmful behavior.
\item \textit{Cheating / feature manipulation}: When individuals may partially mimic compliance, optimal classifiers (with or without a quota constraint) remain simple for a broad and empirically relevant class of signal distributions: they are either monotone single-threshold rules or non-monotone ``two-cut" rules.
\item  \textit{Interpretation}: Our results show that optimal classifiers can depart sharply from those studied in incentive-aware ML and classical screening models, taking the form of negative thresholds or two-cut rules that do not arise in those settings. We also link classification objectives to their behavioral consequences for a population, and characterize a class of objectives that can rule out certain incentive pathologies.
\end{enumerate}

\paragraph{Relation to Incentive-Aware Machine Learning and Performative Prediction.}
Our analysis builds on the growing literature on incentive-aware machine learning, where individuals strategically modify their observable features in response to algorithmic decision rules. Recent survey work \citep{Podimata25} organizes this literature around three themes: robustness to gaming, fairness under heterogeneous capacity for adaptation, and the distinction between gaming and genuine improvement. 

Early papers in the robustness tradition typically model classification as a sequential learner-agent interaction in which people best respond to a classifier by strategically manipulating observable features. A central contribution of this literature is to formalize this interaction as a Stackelberg game in which the classifier commits to a rule by anticipating the best responses of those being classified \citep{HardtEtAl16}. Subsequent work extends this Stackelberg framework, studying incentive-aware and fairness-aware classifier design under feature manipulation \citep{Dong18,HuEtAl19,MilliEtAl19,Chen2018Strategyproof}. In all of these models, agents engage in feature manipulation: the agents' choices alter the signals observed and used by the classifier to make its decisions, but not the agent's underlying type, productivity, or---most importantly---the classification of the individual that is ``correct'' from the standpoint of the classifier's objective.

A complementary ``improvement and causality" line of research instead treats strategic responses as genuine behavioral change \citep{KleinbergRaghavan20,MillerEtAl20,Shavit20}. In these papers ``effort" is always productive in principle, and the designer's challenge is to steer behavior in a desirable direction.  This work is more squarely focused on how the classifier shapes agents' incentives, as opposed to the challenge of ``learning'' the data distribution.  In both senses, our framework shares a lot in common with these earlier contributions.  For example, we too set aside learning issues entirely.  However, a distinction between our framework and much of the work in this vein is that our analysis emphasizes the role of the algorithm designer's uncertainty about the data in shaping its optimal algorithm.

A related literature on performative prediction and outcome performativity \citep{PerdomoZrnicMendlerDunnerHardt2020,KimPerdomo23,HardtMendler-Dunner23} studies learning problems where model deployment changes the data distribution. These models formalize both feature and outcome performativity, whereby predictions reshape the data they aim to predict. Equilibrium is defined by \textit{performative stability}, a fixed point where the model is optimal given the distribution it induces. Our framework differs in both mechanism and equilibrium concept. We micro-found individual responses as strategic choices that determine the induced distribution of types in a population, and we characterize the global Stackelberg optimum rather than a local fixed point. Whereas outcome performativity describes unintentional feedback from model deployment, our analysis treats feedback as intentional and optimized. The designer anticipates behavioral responses to classification and chooses its rule accordingly.

Relative to these literatures, our contribution is to endogenize the population distribution itself and to characterize the globally optimal classification rule. Individual actions change both signals (features) and true type prevalence, generating a nonlinear mapping between the classifier's rule and the population it observes. Because behavioral responses enter the designer's problem nonlinearly, our framework can produce equilibria that the improvement-and-causality and performative stability frameworks cannot capture, including cases in which global accuracy is maximized by negative-threshold or two-cut rules. In short, our framework complements these literatures by deriving the global Stackelberg optimum when the informational environment itself evolves endogenously with the decision rule.

\paragraph{Relation to Signaling and Screening Models in the Social Sciences.}

A parallel tradition in economics and political science has long modeled the behavioral feedbacks of policy rules.  In classical models of signaling and screening \citep{Spence73,RothschildStiglitz76,CrawfordSobel82}, the designer anticipates agents' strategic responses, but does so under highly structured objectives such as profit maximization or efficient risk separation. Models of statistical discrimination  \citep{Phelps72,Arrow73Discrimination} formalized the idea that people strategically invest in costly signals to influence classification outcomes. \cite{CoateLoury93} extended this logic to affirmative-action and hiring thresholds, showing that equilibrium investment responses can sustain or eliminate discriminatory beliefs depending on policy design.  Political science work on monitoring, oversight, and institutional design \citep{McCubbinsSchwartz84,GilliganKrehbiel87,LupiaMcCubbins98} similarly emphasizes how rules shape behavior and decisions.  More recently, \cite{JungEtAl20} bring this approach to algorithmic settings, modeling classification as a screening mechanism that anticipates behavioral responses in a Stackelberg manner but with a particularly tractable designer objective of compliance maximization.  

Our analysis belongs to this broad tradition but departs from it in important ways.  First, we allow the designer to optimize over general classification objectives---including accuracy, precision-like criteria, quotas, and compliance maximization---rather than the more narrow payoff structures of classical screening models. Second, while models of statistical discrimination typically analyze Cournot-style equilibria in which rules and behaviors are mutually consistent given fixed beliefs, and Jung et al. study a special Stackelberg case under compliance maximization, we characterize the global Stackelberg optimum for general designer preferences, fully internalizing both behavioral feedback and resulting changes in the informativeness of equilibrium signals. Our broader treatment yields structural forms for optimal classifiers that do not arise in either screening models or statistical-discrimination frameworks.

Our model also has connections to a more recent literature on information design in political economy (\textit{e.g.}, \cite{KamenicaGentzkow11}, \cite{BergemannMorris19}).  Because we are considering a Stackelberg game, the algorithm in our model is mathematically equivalent to a Blackwell experiment designed ``by a third party'' (here, the algorithm designer) for use by the individual when choosing his or her own behavior.  In line with much of the work on information design, some of our results reflect the impact of the algorithm designer's commitment power: in many cases, we can achieve more efficient social outcomes in this setting than could be achieved if the designer did not have this commitment power.  A subtle point of the results we present below is that the \textit{form} of classifier that the designer can find necessary to achieve its objective can be counterintuitive to the point of appearing pathological.

Finally, \cite{Penn25} studies a related model of classification with endogenous behavior and shows that, when individuals change their labels in response to the classifier, an optimal rule is always a (positive or negative) threshold for a restricted class of ``accuracy-aligned" and ``accuracy-misaligned" designer objectives. The present article develops a substantially more general framework. We extend the analysis to all linear designer payoff functions, introduce quota constraints on the share of positive classifications (Theorem \ref{quotaTheorem}), and provide preference conditions that rule out negative-threshold optimality (Theorem \ref{noNeg}). We further enlarge the behavioral environment by allowing a third action---cheating, in which individuals can ``purchase" a misleading signal---and show that this yields two-cut optimal rules (Theorem \ref{twoCuts}).

\paragraph{Structure of the article.}
Section~\ref{Sec:Baseline} introduces the baseline model. Section~\ref{Sec:Threshold} characterizes the equilibrium incentives induced by any binary classifier and derives the globally optimal classification strategy for a designer, both with and without a quota constraint on positive classifications. Section~\ref{justPositive}  provides preference conditions for the designer under which interior negative-threshold rules are never optimal, ensuring that optimal classification cannot disincentivize compliance. Section~\ref{Sec:Cheat} extends the model to environments with an intermediate, ``cheating" action and shows that---in a wide class of signal environments---the optimal classifier is either a one-cut or a two-cut rule. In Section~\ref{Sec:Examples}, we present several illustrative examples. Section~\ref{Sec:Conclusion} concludes.

\section{Baseline Model \label{Sec:Baseline}}

Consider two players: an individual $i$, and an algorithm, $D$. The individual must choose a binary \textit{behavior}, $\beta_i\in\{0, 1\}$. We refer to $\beta_i=1$ as \textit{compliance} and $\beta_i=0$ as \textit{noncompliance}. If choosing $\beta_i=1$ the individual pays cost $c_i$.\footnote{If $c_i<0$, then it represents a benefit, but we refer to $c_i$ as a cost throughout.} The cost, $c_i$, is private information to the individual.  It is common knowledge, however, that $c_i$ is distributed according to a probability measure with a smooth cumulative distribution function (CDF), $F:\mathbf{R}\rightarrow (0, 1)$, with full support on $\mathbf{R}$ and corresponding probability density function (PDF), $f: \mathbf{R} \to \mathbf{R}_{++}$.

After $i$ chooses behavior $\beta_i \in \{0,1\}$, the algorithm, $D$, observes a \textit{signal}, $s_i \in\mathbf{R}$.  While $\beta_i$ is observed by $i$, but not by $D$, it is common knowledge that $s_i$ is generated according to a behavior-conditional probability distribution with smooth CDF, $G_{\beta_i}: \mathbf{R}\to (0,1)$, with full support on $\mathbf{R}$ and corresponding PDF, $g_{\beta_i}: \mathbf{R} \to \mathbf{R}_{++}$.  Additionally, we assume that $g_1$ and $g_0$ satisfy the strict monotone likelihood ratio property (MLRP) with respect to $s_i$:\footnote{Note that the assumptions that $F$ is continuous and that  $g_1$ and $g_0$ are continuous with full support are stronger than necessary, but simplify the analysis by allowing us to disregard boundary cases.}
\[
\frac{g_1(s)}{g_0(s)} \text{ is strictly increasing in } s. 
\]
Finally, we conserve space by writing $G\equiv (G_0,G_1)$ to denote the pair of behavior-dependent distributions of the signal, $s_i$, and refer to $G$ simply as the \textit{signal distribution}.

After observing $s_i$, the algorithm makes a binary decision for $i$, denoted by $d_i\in\{0, 1\}$, where we refer to $d_i=1$ as $D$ \textit{rewarding} $i$, and $d_i=0$ as $D$ \textit{punishing} $i$. $D$'s algorithm is denoted by 
\[
\delta: \mathbf{R} \to [0,1],
\]
where $\delta(s) \equiv \Pr[d_i=1\mid s]$ is the probability $i$ is rewarded $(d_i=1)$, given signal $s \in \mathbf{R}$.  Throughout, we require that $\delta(s)$ be Lebesgue-integrable and denote the set of all Lebesgue-integrable functions from $\mathbf{R}\to [0,1]$ by $\mathcal{D}$.
 
\paragraph{Timing of Decisions.}  We consider a \textit{Stackelberg game} in which $D$ publicly commits to $\delta$ prior to $i$'s choice of behavior, $\beta_i$. Specifically, the timing of decisions we consider is as follows:
\begin{enumerate}
\item Individual $i$ privately observes behavioral cost $c_i$, drawn according to $F$.
\item With knowledge of $F$ and $G$, $D$ commits to a classification algorithm $\delta$.
\item With knowledge of $c_i$, $F$, $G$, and $\delta$, $i$ chooses behavior $\beta_i\in \{0,1\}$.
\item Based on $i$'s behavior, $\beta_i$, $i$'s signal, $s_i$, is generated according to $G_{\beta_i}$.
\item Based on $s_i$, $i$ is rewarded or punished according to $D$'s announced algorithm, $\delta$.
\item The process concludes, with $i$ and $D$ receiving their payoffs, $u_i(\beta_i,d_i)$ \& $u_D(\beta_i,d_i)$.
\end{enumerate}
With the timing in hand, we now turn to $i$'s behavioral incentives induced by any algorithm $\delta$.

\subsection{Individual Incentives \& Behavior}

If $i$ is rewarded ($d_i=1$), then $i$ receives a reward, $r_1\in \mathbf{R}$. If $i$ is punished ($d_i=0$), $i$ pays a penalty, $r_0\in \mathbf{R}$.  Accordingly, $i$'s net benefit from being rewarded is denoted by $r \equiv r_1-r_0$.  We normalize the model by assuming throughout that $r>0$.

\paragraph{Individual Payoffs.} We assume that individual $i$'s \textit{payoff function}, $u_i$, is linear in $\beta_i$ and $d_i$:\footnote{We will generally omit $c_i$, $r_1$, and $r_0$ and write $i$'s payoff function simply as $u_i(\beta_i,d_i)$.}
\begin{equation}
    \label{Eq:IndividualPayoff}
    u_i(\beta_i,d_i\mid r_1, r_0, c_i) = r_1 \cdot d_i + r_0 \cdot (1-d_i) - c_i \cdot \beta_i.
\end{equation}
Standard arguments imply that $i$'s optimal choice of behavior, given $c_i$, $\delta$, $G_0$, and $G_1$, is 
\begin{equation}
    \label{Eq:OptimalBehaviorRaw}
\beta_i^*(c_i\mid r,\delta,G) = \begin{cases}
    1 & \text{ if } c_i \leq r \cdot \big(\Pr[d_i=1\mid \beta_i=1, \delta, G] - \Pr[d_i=1\mid \beta_i=0, \delta, G]\,\big),\\
    0 & \text{ if } c_i > r \cdot \big(\Pr[d_i=1\mid \beta_i=1, \delta, G] - \Pr[d_i=1\mid \beta_i=0, \delta, G]\,\big).
\end{cases} 
\end{equation}
For any admissible signal structure, $G$, denote $\delta$'s true positive rate (TPR), false positive rate (FPR), true negative rate (TNR), and false negative rate (FNR) as follows:
\begin{eqnarray*}
TPR(\delta\mid G) & \equiv & \Pr[d_i=1\mid \beta_i=1, \delta, G],\\
FPR(\delta\mid G) & \equiv & \Pr[d_i=1\mid \beta_i=0, \delta, G],\\
TNR(\delta\mid G) & \equiv & \Pr[d_i=0\mid \beta_i=0, \delta, G],\\
FNR(\delta\mid G) & \equiv & \Pr[d_i=0\mid \beta_i=1, \delta, G].
\end{eqnarray*}
With this notation, we can re-express Equation \ref{Eq:OptimalBehaviorRaw} as follows:
\begin{equation}
    \label{behavior}
\beta_i^*(c_i\mid r,\delta,G) = \begin{cases}
    1 & \text{ if } c_i \leq r \cdot \big(TPR(\delta\mid G)-FPR(\delta\mid G)\big),\\
    0 & \text{ if } c_i > r \cdot \big(TPR(\delta\mid G)-FPR(\delta\mid G)\big).
\end{cases} 
\end{equation}

\paragraph{Individual $i$'s Incentives under $\delta$.} We describe the difference in $\delta$'s true positive and false positive rates under $G$ as individual $i$'s \textit{incentives} under classifier $\delta$, and denote it by:
\begin{equation}
\label{incentivesDelta} 
\Delta_\delta(G) \equiv TPR(\delta\mid G)-FPR(\delta\mid G),
\end{equation}  
but will omit the conditioning of $\Delta_\delta$ on $G$ when no confusion can result.  Then, in line with the literature on algorithmic fairness, we refer to the probability that individual $i$ chooses $\beta_i$ under $\delta$, given the signal distribution $G$, as individual $i$'s \textit{prevalence}, which our assumptions imply is equal to the following:
\begin{equation}
\pi^*(\delta) \equiv \Pr_F\left[c_i: \beta_i^*(c_i\mid r, \delta, G)=1\right] = F(r\cdot \Delta_\delta).
\end{equation} 
Thus, individual $i$'s prevalence under $\delta$ is the \textit{ex ante} probability that $i$ chooses behavior $\beta_i=1$ when $i$ knows that $i$'s choice of $\beta_i$ will be classified by $\delta$.

\paragraph{The Designer's Payoff.}  We assume $D$ has linear payoffs over $(\beta_i,d_i)$, as shown in Table \ref{Tab:DPayoff}. In words, $A_1, A_0, B_1, B_0$ represent $D$'s marginal value for the frequency of each of the four cells of the confusion matrix in Table \ref{Tab:DPayoff}.  

\vspace{.1in}
\small
\begin{table}[h]
    \centering
    \caption{The Designer's Payoffs \label{Tab:DPayoff}}
    \begin{tabular}{|c||c|c|} \hline
&\multicolumn{2}{c|}{Decision} \\ \hline
Behavior&$d_i=1$&$d_i=0$ \\ \hline \hline

\multirow{2}{*}{$\beta_i=1$}&$A_1$&$A_0$ \\ 
& (True Positive)&(False Negative) \\ \hline
\multirow{2}{*}{$\beta_i=0$}&$B_0$&$B_1$\\ 
& (False Positive) & (True Negative) \\ \hline
 \end{tabular}
\vspace{.5in}
\end{table}
\normalsize

\begin{table}[H]
\centering
\caption{Some Examples of Designer Objectives\label{objectives}}
 \begin{tabular}{|c||c|c|} \multicolumn{3}{c}{Accuracy} \\ \hline
&$d_i=1$&$d_i=0$ \\ \hline \hline

$\beta_i=1$&$\cellcolor{gray!30}A_1=1$ &$A_0=0$ \\ \hline
$\beta_i=0$&$B_0=0$&\cellcolor{gray!30}$B_1=1$ \\ \hline
 \end{tabular} \,
 \begin{tabular}{|c||c|c|}  \multicolumn{3}{c}{Compliance} \\ \hline
&$d_i=1$&$d_i=0$ \\ \hline \hline

$\beta_i=1$&\cellcolor{gray!30} $A_1=1$ &\cellcolor{gray!30}$A_0=1$ \\ \hline
$\beta_i=0$&$B_0=0$&$B_1=0$ \\ \hline
 \end{tabular}
 
 \vspace{.1in}
 
  \begin{tabular}{|c||c|c|} \multicolumn{3}{c}{$p$-Precision} \\ \hline
&$d_i=1$&$d_i=0$ \\ \hline \hline

$\beta_i=1$&$\cellcolor{gray!40}A_1=1$ &$\cellcolor{gray!10}A_0=p$ \\ \hline
$\beta_i=0$&$B_0=0$&\cellcolor{gray!10}$B_1=p$ \\ \hline
 \end{tabular} \,
 \begin{tabular}{|c||c|c|}  \multicolumn{3}{c}{Predatory} \\ \hline
&$d_i=1$&$d_i=0$ \\ \hline \hline

$\beta_i=1$& $A_1=0$ &$A_0=0$ \\ \hline
$\beta_i=0$&$B_0=0$&\cellcolor{gray!30}$B_1=1$ \\ \hline
 \end{tabular}
\end{table}

Table \ref{objectives} illustrates a few of the different objectives the designer of a classifier could seek to optimize.  An accuracy-maximizing designer seeks to maximize the probability $i$ is classified correctly, but has no intrinsic preference for one type of behavior over another. A compliance-maximizing designer maximizes the probability $i$ chooses $\beta_i=1$. A $p-$precision  designer (for $p\in(0, 1)$) seeks to assign the positive classification outcome to $i$ if and only if the probability $i$ has complied is greater than $p$. A predatory designer wants to maximize the probability that $i$ does not comply and is assigned the negative classification outcome.  More formally, an algorithm $\delta$ is \textit{optimal} for $D$ if $\delta$ maximizes the following: 
\begin{equation}
\label{dUtility}
u_D(\delta)=\pi^*(\delta) \cdot \big(TPR(\delta) \cdot A_1 + FNR(\delta)\cdot A_0\big)+(1-\pi^*(\delta))\big(FPR(\delta) \cdot B_0+ TNR(\delta)\cdot B_1\big).
\end{equation} 
Thus, $D$'s optimal choice of $\delta$ depends both on its effect on prevalence, $\pi^*(\delta)$, and the conditional distribution of $d_i$ given $\beta_i$ (\textit{i.e.}, the \textit{accuracy} of $\delta$'s assignment, $d_i$, with respect to $i$'s behavior, $\beta_i$).  With optimality defined, we now turn to characterizing optimal classification algorithms.

\subsection{Threshold \& Two-Cut Classifiers}

Our results below allow provide sufficient conditions for the designer's optimal algorithm to be either a \textit{threshold rule} or a \textit{two-cut rule}.  We now formally define these families of rules.

\paragraph{Threshold Rules.} Threshold rules are commonly used in practice, and frequently optimal in well-behaved settings.  While all threshold rules are monotonic, this monotonicity can be increasing or decreasing in the score $s$.  The following formally defines both families of threshold rules. 
\begin{definition}
Algorithm $\delta$ is a \textbf{threshold rule} if, for some $\tau \in \mathbf{R}\cup\{\pm\infty\}$, $\delta$ is of one of the following forms:
\begin{eqnarray*}
\text{\emph{Positive $\tau$-Threshold Rule:}} \;\;\;\;\;\; \delta_\tau^+(s) & = & \begin{cases}
    0 & \text{ if } s < \tau, \\
    1 & \text{ if } s \geq \tau,
\end{cases} \\
\text{\emph{Negative $\tau$-Threshold Rule:}} \;\;\;\;\;\; \delta_\tau^-(s) & = & \begin{cases}
    1 & \text{ if } s \leq \tau, \\
    0 & \text{ if } s > \tau,
\end{cases}
\end{eqnarray*}
\end{definition}
for some $\tau \in \mathbf{R}\cup\{\pm\infty\}$.  In words, the positive threshold rule $\delta_\tau^+$ rewards individual $i$ if and only if $i$'s score, $s_i$, is no \textit{less} than $\tau$.  Similarly, the negative threshold rule $\delta_\tau^-$ rewards  individual $i$ if and only if $i$'s score, $s_i$, is no \textit{greater} than $\tau$. Letting $\tau=\infty$ corresponds to assigning $\delta_\infty^+(s)=0$ and $\delta_\infty^-(s)=1$ for all $s$. 

\paragraph{Two-Cut Rules.} In addition to standard threshold rules, we also consider classifiers that assign positive or negative outcomes on a band of scores. We refer to these as ``two-cut rules." As with threshold rules, they come in two canonical forms, defined below.
\begin{definition}\label{twoCutDef}
Algorithm $\delta$ is a \textbf{two-cut rule} if, for some $\tau_L<\tau_H$, $\delta$ is of one of the following forms:
\begin{eqnarray*}
\text{\emph{Inner $(\tau_L,\tau_R)$-Two-Cut Rule:}} \;\;\;\;\;\; \delta^I(s) & = & \begin{cases}
1 & \text{ if }\,\,\, s\in [\tau_L, \tau_H],\\
0 & \text{ otherwise,}
\end{cases} \;\;\;\;\; \text{ or}\\
\text{\emph{Outer $(\tau_L,\tau_R)$-Two-Cut Rule:}} \;\;\;\;\;\; \delta^O(s) & = & \begin{cases}
0 & \text{ if }\,\,\, s\in (\tau_L, \tau_H),\\
1 & \text{ otherwise,}
\end{cases}
\end{eqnarray*}
\end{definition}

\begin{figure}[h!]
\captionsetup{justification=centering}
    \centering
    \framebox{\includegraphics[width=.66\linewidth]{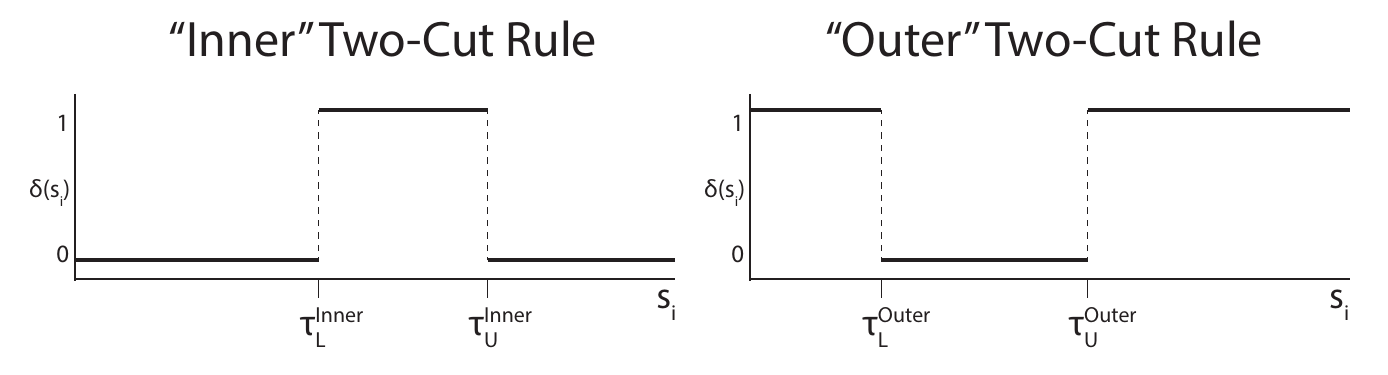}}
    \caption{Examples of Inner and Outer Two-Cut Rules}
    \label{Fig:TwoCutRules}
\end{figure}
\mbox{}\\

As displayed in Figure \ref{Fig:TwoCutRules}, an inner two-cut rule is a ``Goldilocks'' rule (\textit{i.e.}, to be rewarded under $\delta^I$, $s_i$ must be neither too large ($s_i \leq \tau_R$) nor too small ($s_i\geq \tau_L$).  Conversely, an outer two-cut rule is an anti-Goldilocks rule: to be rewarded under $\delta^O$, $s_i$ must be sufficiently small ($s_i\leq \tau_L$) or sufficiently large ($s_i \geq \tau_R$).  

\section{Quotas and Threshold Structure \label{Sec:Threshold}}

Throughout, we investigate environments where the classifier is potentially constrained by a quota on the fraction of individuals who may be rewarded. This allows us to capture settings such as admissions and hiring, where capacity constraints shape classification design.  Specifically, suppose there is a requirement that the expected proportion of individuals who are rewarded ($d_i=1$) by $\delta$ be no greater than $q \in (0,1]$, where $q$ is exogenous and fixed.  Setting $q=1$ is equivalent to $D$ being unconstrained in the share of expected positive classifications that are allowable, and so this environment allows us to capture both constrained and unconstrained settings.

For any admissible $(F,G)$ and $\delta\in \mathcal{D}$, let $v(TPR(\delta), FPR(\delta))$ denote the fraction of individuals rewarded in expectation by $\delta$, defined as follows: 
\begin{equation}
\label{quotaEq}
v(TPR(\delta), FPR(\delta))=\pi^*(\delta) \cdot TPR(\delta) + (1-\pi^*(\delta)) \cdot FPR(\delta).
\end{equation} 
Then, the quota constraint requires that: 
\[
v(TPR(\delta), FPR(\delta)) \leq q.
\]
We now show that the optimal classifier must be a positive or negative threshold rule (Theorem \ref{quotaTheorem}).  Our proof strategy is to discretize our signal space ($\mathbf{R}$) by partitioning it into a finite set of intervals (or ``bins'').  Within this coarsened signal space, the optimal classifier in this discrete problem will be a finite vector that assigns each bin  a probability that $d_i=1$.  We first show that the optimal classifier for this finite binned problem is a threshold- or negative-threshold rule, with at most one bin involving non-degenerate randomization. We then prove that the optimal binned classifier converges to either a positive or negative threshold rule in the continuum.

\begin{theorem}
\label{quotaTheorem} 
Suppose $\delta$ can reward no more than $q\in(0, 1]$ of the individuals. The optimal classifier, $\delta_q^*$, is a threshold or negative threshold rule.
\end{theorem}
\begin{proof}
Proofs of all numbered results are contained in Appendix \ref{appProofs}.
\end{proof}
\noindent\textit{Why Theorem \ref{quotaTheorem} matters.}  Whether or not a quota constraint is present, the optimal classifier---despite behavioral responses to the rule---is still both intuitive and describable by a single number: it is always a (positive or negative) threshold rule. This result shows that allowing behavior to respond strategically to classification does not require complex, non-monotone, or randomizing mechanisms. The surprising feature is that a negative threshold can be accuracy-maximizing, reversing the usual association between higher signals and favorable outcomes. We provide an example of this phenomenon in Section \ref{Sec:Examples}.\footnote{See also \cite{Penn25,PennPatty25AlgEndog}.}  From a machine learning perspective, this illustrates that globally optimal rules can violate familiar monotonicity and calibration intuitions. From a game theoretic perspective, it provides a complete characterization of Stackelberg-optimal screening when the population's behavior and the distribution of observed signals are jointly determined by the rule.

\subsection{ \label{justPositive} Designer Preferences That Rule Out Negative Thresholds}

Theorem \ref{quotaTheorem} highlights a potential pathology of optimal classifiers: it can be optimal for a designer to \textit{reduce} compliance in order to maximize its objective. While this is not surprising when the designer explicitly prefers low compliance, the pathology extends even to seemingly ``neutral" objectives such as accuracy maximization. In such cases, an accuracy-maximizing designer may amplify behavioral disparities across groups, incentivizing strictly greater compliance in one population (via a positive threshold rule) while strictly disincentivizing compliance in another (via a negative threshold rule).

Related observations appear in prior work on strategic manipulation.
\cite{HuEtAl19}  show that accuracy-maximizing classifiers can generate opposite incentive effects across groups when  manipulation costs differ, and  \cite{MilliEtAl19} show that accuracy-oriented design can distort incentives in socially costly ways. The similarities between our results and those of \cite{HuEtAl19} and \cite{MilliEtAl19} are particularly interesting because the reason they arise in our setting is logically independent from why they emerge in their frameworks.  Specifically, our agents' behavioral responses do not simply ``complicate'' the algorithm designer's optimization problem: these responses actually shape the algorithm designer's first order incentives.  In \cite{HuEtAl19} and \cite{MilliEtAl19}, each individual has a correct ``label'' (\textit{i.e.}, a correct classification) that is exogenous and unaffected by the individual's behavioral response to the algorithm.  That is, individuals do not change their underlying type. Strategic behavior consists of ``gaming the system" by manipulating observable features. In our setting, group-level incentive reversals similarly require heterogeneity in either cost or signal distributions, but the underlying pathology is potentially more severe. Optimal classifiers may induce real behavioral change, not just feature manipulation, and may do so in a direction that reduces socially desirable behavior.

The following result, Theorem \ref{noNeg}, provides a set of simple order restrictions on designer preferences $(A_1, A_0, B_1, B_0)$ under which an interior negative threshold rule---the class of rules strictly disincentivizing compliance---is never optimal for the designer. In this sense, the theorem identifies preference structures that ensure optimal classification remains compliance-preserving. We then turn immediately to Proposition \ref{tightProp}, which shows that these conditions are tight.

\begin{theorem}\label{noNeg} Let $[A_1\geq B_0]$, and $[A_0\geq B_1]$ and $[A_1\geq A_0$ or $B_1\geq B_0]$. Then an interior negative threshold rule is never optimal for the designer.
\end{theorem}

\noindent\textit{Intuition for Theorem \ref{noNeg}.} The inequalities in Theorem \ref{noNeg} have a natural interpretation. The first pair, $$A_1\geq B_0\,\,\,\,\,\text{ and }\,\,\,\,A_0\geq B_1,$$ express \textit{compliance monotonicity}: conditional on any decision, the designer weakly prefers that an individual be compliant rather than noncompliant. The second requirement,
$$A_1\geq A_0\,\,\,\,\,\text{ and }\,\,\,\,\,B_1\geq B_0,$$ is an \textit{accuracy-alignment} condition that prevents the designer from strictly preferring misclassification in both behavioral states. Together, these preference restrictions ensure that the designer never finds it optimal to strictly depress compliance through a negative threshold.

As promised above, Proposition \ref{tightProp} shows that none of the inequalities in Theorem \ref{noNeg} can be weakened within the class of simple order restrictions on $A_1, A_0, B_1, B_0$ without admitting an environment in which an interior negative threshold rule is optimal. We prove this result by numerically constructing three environments (a choice of $A_1, A_0, B_1, B_0$, signal distributions $G_1$ and $G_0$, net reward $r$, and cost distribution $H$) in which only one inequality is violated, and in which an interior negative threshold rule is strictly optimal for the designer.

\begin{proposition}\label{tightProp} If any of the inequalities in Theorem \ref{noNeg} fails, then there exists an environment in which an interior negative threshold is optimal.
\end{proposition}

\noindent\textit{Why Theorem \ref{noNeg} and Proposition \ref{tightProp} matter.}   These results provide a clean link between a designer's classification objective and the direction of behavioral change induced in equilibrium.  A useful benchmark to consider is the $p$-precision designer objective defined in Table \ref{objectives}. A $p-$precision designer seeks to issue a positive classification only when it is at least $p$ likely to be correct, effectively imposing a uniform lower bound on the positive predictive value (PPV) of the classifier.  Although accuracy and $p$-precision may appear similar in spirit, only the latter satisfies the conditions of Theorem \ref{noNeg}; $p$-precision designers will never optimally induce opposing incentive effects across groups or disincentivize compliance. This highlights the importance of aligning the designer's objective with its equilibrium consequences. Modest adjustments to a designer's objective can eliminate the negative-threshold pathology entirely.  From a machine learning perspective, this identifies a simple modification of classification objectives that rules out equilibrium incentive reversals and ensures monotone behavioral responses. From a game theoretic perspective, we've characterized when Stackelberg-optimal screening mechanisms must be compliance-promoting, linking a designer's payoff weights to the qualitative direction of induced behavior.

\section{Intermediate Behaviors and Two--Cut Optimality: \\Introducing ``Cheating'' \label{Sec:Cheat}}

We now allow individuals to choose an intermediate action, which we refer to as \textit{cheating}: a lower-cost behavior that partially mimics compliance by improving an individual's signal without altering their true type. This extension captures environments in which individuals can strategically manipulate features or partially invest in ways that increase their likelihood of classification without satisfying the underlying criterion of ``compliance."

Formally, we extend the model of Section~\ref{Sec:Threshold} by giving each individual a third behavioral option. Cheaters are treated as noncompliant in the designer's payoff function; the designer does not directly value cheating relative to noncompliance. However, cheating matters indirectly through at least two channels:
\begin{itemize}
\item \textit{``Direct Effect'' of Cheating.} Because cheating is cheaper than full compliance, some individuals who would have complied in the baseline model now prefer to cheat.
\item \textit{``Indirect Effect'' of Cheating.} Cheating ``muddies'' the score distribution of noncompliant types, reducing the classifier's ability to distinguish compliers from noncompliers and cheaters.
\end{itemize}

Both channels change the designer's incentives. For designers who wish to maximize compliance, the presence of cheating may induce the use of a more restrictive positive threshold to discourage individuals from switching away from full compliance. For accuracy-oriented designers, the second channel can make it optimal to adopt a \emph{two-cut} rule that rewards or punishes only an intermediate band of signals. As with threshold rules, optimal two-cut rules can induce or deter compliance, and examples of both types of rules are presented in Section~\ref{Sec:Examples}. Importantly, two-cut rules that deter cheating induce individuals to self-sort into compliant or noncompliant behaviors. This self-selection restores separation in the signal distribution and improves the informational environment faced by the classifier. This mechanism is illustrated concretely in Example \ref{twoCutExOuter}.

\paragraph{Primitives.}  The individual's behavioral choices are now denoted by $\mathcal{B}\equiv\{0,1,\chi\}$, where $\beta_i=0$ represents $i$ \textit{not complying} and $\beta_i=1$ represents $i$ \textit{complying}, as before. If $i$ chooses $\beta_i \in \{0,1\}$, the remainder of the process is identical to that in the baseline model.  If $i$ chooses to cheat (\textit{i.e.}, $\beta_i=\chi$), then the signal associated with $i$, $s_i$, is generated by a smooth probability distribution with full support on $\mathbf{R}$, and we denote the CDF \& PDF of this distribution by $G_\chi: \mathbf{R}\to (0,1)$ and $g_\chi: \mathbf{R}\to \mathbf{R}_{++}$, respectively.

As in the baseline model, $i$ privately observes his or her cost of compliance, $c_i$, prior to choosing behavior (which in this extension is now $\beta_i \in \{0,1,\chi\}$).  Individual $i$'s payoff function is extended to this environment as follows:
\[
u_i(\beta_i,d_i\mid c_i, r, \kappa) = \begin{cases}
    r_1 \cdot d_i +r_0 \cdot (1-d_i)& \text{ if } \beta_i=0,\\
     r_1 \cdot d_i +r_0 \cdot (1-d_i)- c_i & \text{ if } \beta_i=1,\\
     r_1 \cdot d_i +r_0 \cdot (1-d_i) - \kappa \cdot c_i & \text{ if } \beta_i = \chi,
\end{cases}
\]
where $\kappa \in [0,1]$ is an exogenous and commonly known parameter representing ``how expensive cheating is.''\footnote{We assume that cheating is differentially costly for individuals and positively correlated with their true cost of compliance, $\mathrm{c_i}$. If individuals faced a common and fixed cost to cheating, $c_\chi$, then we would not be able to sustain an equilibrium in which compliance, noncompliance, and cheating behavior could all be observed.}

\paragraph{Informational Structure.}  As in the baseline model, we assume that $g_1$ satisfies the MLRP with respect to $g_0$.  In this extension we additionally assume that $g_\chi$ satisfies the MLRP with respect to $g_0$, so that expected scores are higher conditional on cheating than on noncompliance. Our results that follow place no restriction on the MLRP ordering between cheating and compliance: either $g_1$ may dominate $g_\chi$ or vice versa. However, the non-monotonic (``two cut") structure of an optimal rule is most surprising when compliant behavior shifts the score distribution upward relative to cheating, and it is this case that we illustrate in Section \ref{Sec:Examples}. If cheating is associated with higher scores than compliance then two-cut rules arise in a straightforward way, since both very low and very high scores are  indicative of noncompliance (the latter via cheating).

\ \\ The introduction of cheaters introduces some subtleties into our problem when we allow behavior to respond to classification. For each $\beta \in \{0,1,\chi\}$, define $S_\beta(\delta)$ to be the conditional probability $i$ is rewarded ($d_i=1$) by $\delta$, if $i$ chooses behavior $\beta_i = \beta$:
 
\[
S_\beta(\delta)=\int_{\mathbf{R}}g_\beta(s) \cdot \delta(s) \; \dee s.
\]
Then, $i$'s expected payoffs for each behavior, $\beta_i \in \{0,1,\chi\}$, conditional on $r$, $\delta$, $c_i$, and $\kappa$, are:
\begin{equation}\label{payoffsCheating}
EU(\beta_i|\delta)=\left\{\begin{array}{ll} 
r \cdot S_1(\delta) +r_0 - c_i&\text{ if }\beta_i=1,\\
r\cdot S_\chi(\delta) +r_0 - \kappa \cdot c_i&\text{ if }\beta_i=\chi,\\
r\cdot S_0(\delta) +r_0 &\text{ if }\beta_i=0.
\end{array}\right.
\end{equation}
To simplify notation, define the following:
\begin{equation}
\label{Eq:BigDeltaCheating}
\begin{array}{rcl}
\Delta^{1\chi}_{\delta} & \equiv & S_1(\delta) -S_\chi(\delta) ,\\[0.2em]
\Delta^{10}_{\delta} & \equiv & S_1(\delta) -S_0(\delta) , \\[0.2em]
\Delta^{\chi 0}_{\delta} & \equiv & S_\chi(\delta) -S_0(\delta) .
\end{array}
\end{equation}
The quantity $\Delta^{ab}_\delta$ equals the difference in the probability of $d_i=1$ when choosing behavior $\beta_i=a$ rather than $\beta_i=b$, given $\delta$. Accordingly $i$'s optimal behavior is: 
\begin{equation}
\label{behavior3}
\beta_i^*(c_i \mid r,\delta,\kappa) \; = \; \begin{cases}
1 & \text{ if }\,\,\, c_i\leq \min\left\{r\cdot \Delta^{10}_\delta, \;\; r\cdot \dfrac{\Delta^{1\chi}_\delta}{1-\kappa} \right\},\\[1em]
\chi & \text{ if }\,\,\, c_i\in \left( r\cdot \dfrac{\Delta^{1\chi}_\delta}{1-\kappa} , \;\; r \cdot \dfrac{\Delta^{\chi 0}_\delta}{\kappa} \right),\,\,\, \\[1em]
0 & \text{ if }\,\,\, c_i\geq \max \left\{r\cdot \Delta^{10}_\delta, \;\; r\cdot \dfrac{\Delta^{\chi 0}_\delta}{\kappa} \right\}. \end{cases}
\end{equation}

\ \\ Equation \ref{behavior3} implies that there are two possible types of aggregate behavior in equilibrium. No cheating occurs in equilibrium if:
\[
\kappa \geq  \dfrac{\Delta^{\chi 0}_\delta}{\Delta^{\chi 0}_\delta+\Delta^{1\chi}_\delta} 
=\dfrac{\Delta_\delta^{\chi 0}}{\Delta_\delta^{10}}.
\]
Otherwise, some individuals will optimally cheat.\\~\\
The effect of cheating on the designer's objective function is captured by the equilibrium fraction of compliers, non-compliers and cheaters in the population. Respectively, these terms are:
\begin{equation}
\label{cheatPrevalence}
\begin{array}{rcl} 
F^1_\delta& \equiv & F\left(\min\{r\cdot \Delta^{10}_\delta, r\cdot \dfrac{\Delta^{1\chi}_\delta}{1-c}\}\right),\\[1em]
F^0_\delta & \equiv & 1-F\left(\max\{r\cdot \Delta^{10}_\delta, r\cdot \dfrac{\Delta^{\chi 0}_\delta}{c}\}\right),\\[1em]
F^\chi_\delta & \equiv & 1-F^1_\delta-F^0_\delta.
\end{array}
\end{equation}
Incorporating the possibility of cheating implies that the designer's objective is now: 
\begin{equation}\label{Cheat}
u_D(\delta)=F^1_\delta\Big(S_1\cdot A_1+(1-S_1)\cdot A_0 \Big)+ 
F^\chi_\delta\Big(S_\chi\cdot B_0+(1-S_\chi)\cdot B_1\Big) + 
F^0_\delta \Big(S_0\cdot B_0+(1-S_0)\cdot B_1\Big).
\end{equation}
  Note that $F^1_\delta$, $F^0_\delta$, and $F^\chi_\delta$ are continuous in $\delta$ and directionally differentiable, but not necessarily smooth. We now extend the proof strategy of Theorem \ref{quotaTheorem} to consider classification with cheating, and show that optimal classification is no longer necessarily characterized by a single threshold. Imposing the following additional condition on our signal distributions enables us to obtain a clean characterization of optimal classification with cheating.
  
  \begin{assumption}We assume that our signal distributions share the common exponential form $$g_\beta(s)=h(s)e^{a_\beta+b_\beta T(s)}$$ with $T(s)$ strictly monotone and differentiable, $h(s)>0$ and continuous, $a_\beta, b_\beta\in\mathbf{R}$, and $b_\beta\not=b_{\beta^\prime}$ for $\beta\not=\beta^\prime$.
\label{twoCutAss}
\end{assumption}

\noindent \textit{Remark.}  The proof of Theorem~\ref{twoCuts} relies on a ``two-crossing'' property of a linear combination of the densities $\{g_\beta\}$. Assumption~\ref{twoCutAss} provides a convenient sufficient condition for this property to hold and is satisfied by many natural and commonly used signal families including the Normal distribution. However, the two-crossing property we need is not unique to exponential-family densities. A fuller discussion of the assumption, including its relation to other signal environments, is provided in Appendix~\ref{assDis}.

\ \\ Theorem \ref{twoCuts} shows that when cheating is feasible, the optimal classifier---whether or not the designer faces a quota constraint---has at most two cutpoints under Assumption \ref{twoCutAss}. Thus, the optimal rule may no longer be characterized by a single threshold $\tau$, but the classification boundary remains low-dimensional: one cutpoint (a threshold rule) or two cutpoints (a two-cut rule). In other words, allowing cheating enlarges the set of candidate optimal rules, but does not lead to complex rules.

\begin{theorem} 
\label{twoCuts}
Suppose that $g_\beta$ share a common exponential form, as defined in Assumption \ref{twoCutAss}. The optimal classifier satisfying quota $q \in (0,1]$ is either a threshold rule or a two-cut rule.
 \end{theorem}
 
\noindent\textit{Why Theorem \ref{twoCuts} matters.} When individuals can ``cheat" by partially imitating compliance at lower cost, the optimal classifier may reward only a middle band (or its complement), producing a two-cut rule. This breaks the one-threshold structure common to standard strategic-classification and mechanism-design models. The result shows that optimal screening can, for example, deliberately make intermediate signals unattractive to improve the informational value of extreme behaviors. For CS/ML readers, this identifies a tractable setting where the optimal decision boundary is non-monotone yet fully characterized; for game theorists, it highlights how equilibrium incentives can endogenously generate pooling and exclusion regions reminiscent of non-monotone equilibria in signaling models.

\section{Illustrative Examples \label{Sec:Examples}}

Before concluding, we offer three simple examples in our setting in which global accuracy is maximized by a negative threshold or two-cut rule.  The first example (Example \ref{negCut}) considers a setting, without cheating, in which the global accuracy-maximizing algorithm is a negative threshold rule.\footnote{\cite{PennPatty25AlgEndog} recently showed such a rule to be optimal in a simpler setting with binary signals.}  We denote the CDF of the $\mathrm{Normal}(\mu,\sigma^2)$ distribution by $\Phi_{\mu,\sigma^2}$ and the PDF by $\phi_{\mu,\sigma^2}$.

\ \\\noindent \textbf{Example \exampleNum{negCut} \theexample .} \textit{Global Accuracy Maximized by a Negative Threshold Rule.}\\~\\
Suppose that: 
\begin{eqnarray*}
r & = & 4,\\
c_i & \sim & \mathrm{Normal}(1/2,1/2), \text{ and}\\
s_i \mid \beta_i & \sim & \mathrm{Normal}(\beta_i,1).
\end{eqnarray*}
In this setting, the ``sincere" level of compliance---the fraction of the population that would comply in the absence of any extrinsic incentives---is $F(0) = \Phi_{0.5,0.5}(0) = 16\%$.  Table \ref{Tab:NegativeThresholdGlobalAccuracy} reports the most accurate positive and negative threshold rules in this setting. 
\begin{table}[h!]
\centering
\begin{tabular}{|c|c||c|} \cline{2-3}
\multicolumn{1}{c|}{} &\textbf{Most Accurate} &\textbf{Most Accurate} \\ 
\multicolumn{1}{c|}{} &\textbf{Positive Threshold} &\textbf{Negative Threshold} \\ 

\multicolumn{1}{c|}{} & ($\tau^*=-0.26$) & ($\tau^*=-1.55$) \\ 
\hline
\textbf{Accuracy} 	&85\%&87\%\\ \hline
\textbf{Compliance} 	&91\%&7\%\\ \hline
\end{tabular}
\caption{Comparing Maximal Accuracy of Positive \& Negative Threshold Rules \label{Tab:NegativeThresholdGlobalAccuracy}}
\end{table}\\
In this setting, the most accurate positive threshold rule rewards signals above $\tau^*=-0.26$, and the most accurate negative threshold rule rewards signals \textit{below} $\tau^*=-1.55$.  As illustrated in Table \ref{Tab:NegativeThresholdGlobalAccuracy}, the most accurate threshold rule in this setting is the negative threshold rule.  By Theorem \ref{quotaTheorem}, this negative threshold rule is the most accurate classification algorithm in this setting.

Some intuition behind this conclusion is as follows.  Accuracy maximization depends on both the precision of the signal received by the algorithm (which depends on signal distribution $G$) \textit{and}  the base rate of compliance in a population (which depends on $F$). Importantly, the most challenging situation for an accuracy-maximizing classifier occurs when the prevalence of compliance in a population  is exactly 50\%.  Such a situation maximizes the uncertainty faced by a classifier regarding behavior $\beta_i$, because $D$ must rely solely on signal information $s_i$ in order to render a decision.  To provide some more insight into the differences between the two rules described in Table \ref{Tab:NegativeThresholdGlobalAccuracy}, the expected distribution of their outcomes are represented in confusion matrix form in Table \ref{Tab:NegativeThresholdGlobalAccuracyConfusionMatrix}.
\begin{table}[h!]
\centering
\begin{tabular}{|c|c|c||c|c|} \cline{2-5}
\multicolumn{1}{c|}{}  & \multicolumn{2}{|c||}{\textbf{Positive Threshold Rule}} & \multicolumn{2}{c|}{\textbf{Negative Threshold Rule}}\\ \cline{2-5}
\multicolumn{1}{c|}{} &$d=1$&$d=0$& $d=1$&$d=0$\\ \hline
$\beta=1$& 82\% \,\,(TP) &9 \%\,\,(FN)& $\approx$0\% \,\,(TP) &7 \%\,\,(FN)\\  \hline
$\beta=0$ &5\%\,\,(FP)& 3\%\,\,(TN)&6\%\,\,(FP)& 87\%\,\,(TN)\\ \hline
\end{tabular}
\caption{Comparing Positive \& Negative Threshold Confusion Matrices \label{Tab:NegativeThresholdGlobalAccuracyConfusionMatrix}}
\end{table}

 An interior positive threshold rule will always strictly \textit{increase} equilibrium compliance relative to sincere compliance, while an interior negative threshold rule will always strictly \textit{decrease} equilibrium compliance relative to sincere compliance.  If sincere compliance is relatively low, as it is here at 16\%, it may be optimal to commit to a negative threshold in order to most productively manipulate base rates. This type of rule commits to classifying extreme signals incorrectly, but the resulting base rate of 7\% enables the  classifier to improve global accuracy in the face of noisy signal information.

\ \\ The final two examples allow individuals to cheat ($\beta_i = \chi$) and demonstrate that each type of two-cut rule can be optimal.  Example \ref{twoCutExOuter} considers a setting where cheating ($\beta_i=\chi$) is difficult to distinguish from true compliance ($\beta_i=1$).  

\ \\\noindent \textbf{Example \exampleNum{twoCutExOuter} \theexample .} \textit{Global Accuracy Improved by an Outer Two-Cut Rule.}\\~\\
Suppose that cheating ($\beta_i=\chi$) is possible, and the parameters of the setting are as follows:
\begin{eqnarray*}
r & = & 3,\\
c_i & \sim & \mathrm{Normal}(0,1),\\
\kappa & = & 0.4, \;\;\;\;\;\;\;\;\;\;\;\;\;\;\;\;\;\;\;\;\;\;\;\;\;\;\;\;\;\;\;\;\;\;\;\;\;\;\;\;\;\;\;\;\; \text{ and}\\
s_i \mid \beta_i & \sim & \begin{cases}
\mathrm{Normal}(0,1) & \text{ if }\beta = 0,\\
\mathrm{Normal}(1.75,1) & \text{ if }\beta = \chi,\\
\mathrm{Normal}(2,1) & \text{ if }\beta = 1.
\end{cases}
\end{eqnarray*}
In this setting, the sincere level of compliance is $F(0) = \Phi_{0,1}(0)=50\%$.  
Relative to our population in Example \ref{negCut}, this population faces, on average, lower costs to compliance. Moreover, cheating is potentially very effective here, because the distribution of signals conditional on cheating (\textit{i.e.}, $s_i \mid \beta_i=\chi$) is ``closer'' to the distribution of signals conditional on compliance (\textit{i.e.}, $s_i \mid \beta_i=1$) than it is to the distribution of signals conditional on non-compliance (\textit{i.e.}, $s_i \mid \beta_i=0$).

\begin{figure}[h!]
\centering
\framebox{
\begin{tabular}{cc}
\includegraphics[scale=0.75]{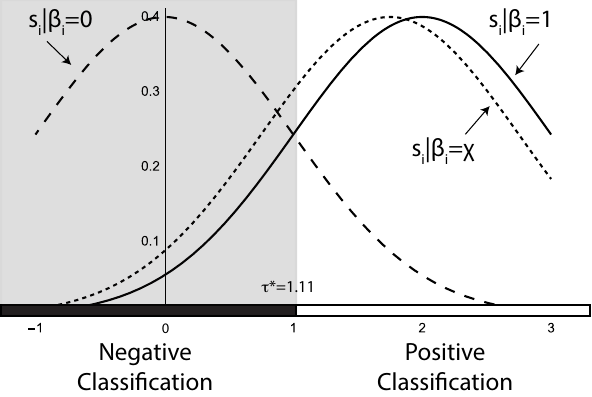} & \includegraphics[scale=0.75]{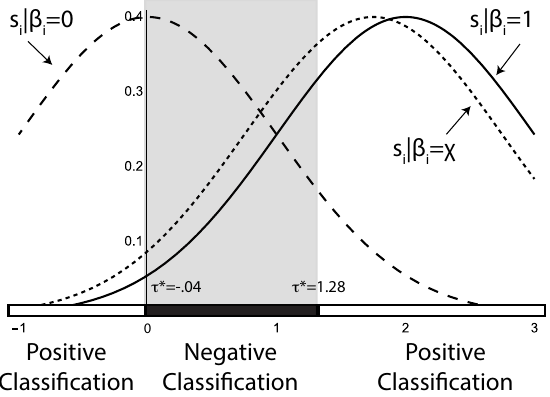}\\
\textbf{Optimal one-cut rule} & \textbf{Better two-cut rule}
\end{tabular}
}
\caption{Comparing Threshold and Outer Two-Cut Rules \label{Fig:TwoRulesWithCheatingOuter}}
\end{figure}

For these parameter values, the optimal threshold rule sets a positive threshold equal to $\tau^*=1.11$.\footnote{For clarity, we do not report the most accurate negative threshold rule in this setting, as it is out-performed by the positive threshold.} We find a more accurate two-cut rule  assigns the negative classification outcome only to those individuals whose signals fall on the interval $(-.04, 1.28)$.  In this example and the next we say that our two-cut rules are ``more accurate" but not ``optimal" simply because we cannot verify analytically that this two-cut rule achieves the global optimum across all possible two-cut configurations. The two-cut rules we present dominated all other two-cut configurations (including both inner and outer two-cut rules) in our calculations.  Importantly, however, we \textit{can} verify that these rules strictly dominate every single-threshold rule under the same primitives. Accordingly, these examples demonstrate that no monotone threshold rule attains the designer's optimum in this environment.  Figure \ref{Fig:TwoRulesWithCheatingOuter} displays the two rules, along with the three $\beta_i$-conditional signal distributions.\footnote{The curves correspond to the conditional signal densities (\textit{i.e.}, PDFs) for compliers (solid), cheaters (small-dash), and non-compliers (large-dash).}

Table \ref{Tab:OuterTwoGlobalAccuracy} reports the accuracy, compliance, and probability of cheating under both rules. It shows that the two classifiers produce similar levels of accuracy (with the two-cut rule representing a slight improvement) and similar levels of compliance (with the positive threshold rule inducing slightly more compliance). However, the two-cut rule reduces cheating behavior by about $14\%$.

\begin{table}[h!]
\centering
\begin{tabular}{|c|c||c|} \cline{2-3}
\multicolumn{1}{c|}{} & \textbf{Most Accurate} &\textbf{More Accurate} \\ 
\multicolumn{1}{c|}{} & \textbf{Positive Threshold} &\textbf{Outer Two-Cut Rule} \\ 
\multicolumn{1}{c|}{} & $\tau^*=1.11$&$(\tau^*_L, \tau^*_H)=(-.04, 1.28)$\\ \hline
\textbf{Accuracy} &61\%& 62\%\\ \hline
\textbf{Compliance} &65\%&63\% \\ \hline
\textbf{Cheating} &35\%& 21\%\\ \hline
\end{tabular}
\caption{Comparing the Accuracy of Positive \& Outer Two-Cut Rules \label{Tab:OuterTwoGlobalAccuracy}}
\end{table}

The intuition for this finding is that a two-cut rule induces individuals to self-sort away from cheating behavior: by committing to positively classify some low signals, some people that might previously have been motivated to cheat are less motivated to do so. The classifier consequently incentivizes people to generate a distribution of data that it can more easily classify accurately.  Of course, the downside  is that by making the data easier to sort, the classifier must commit to sorting some signals incorrectly. In this case the trade-off is, at the margin, worth it.

\begin{table}[h!]
\centering
\begin{tabular}{|c|c|c||c|c|} \cline{2-5}
\multicolumn{1}{c|}{}  & \multicolumn{2}{|c||}{\textbf{Positive Threshold Rule}} & \multicolumn{2}{c|}{\textbf{Outer Two-Cut Rule}}\\ \cline{2-5}
\multicolumn{1}{c|}{} &$d=1$&$d=0$& $d=1$&$d=0$\\ \hline
$\beta=1$& 52\% \,\,(TP) &12\%\,\,(FN)& 49\% \,\,(TP) &14\%\,\,(FN)\\  \hline
$\beta\neq 1$ &26\%\,\,(FP)& 9\%\,\,(TN)&22\%\,\,(FP)& 13\%\,\,(TN)\\ \hline
\end{tabular}
\caption{Comparing Positive Threshold \& Outer Two-Cut Confusion Matrices \label{Tab:OuterTwoCutRuleGlobalAccuracyConfusionMatrix}}
\end{table}
Table \ref{Tab:OuterTwoCutRuleGlobalAccuracyConfusionMatrix} reports the classification outcome distributions under each rule and  provides another illustration of the similar performance of the most accurate threshold rule and the better  two-cut rule in this setting.  Although their overall performance is similar, the underlying decision boundaries differ markedly, as Figure \ref{Fig:TwoRulesWithCheatingOuter} makes clear. 

To conclude our examples, we now present a setting where the designer benefits from choosing a classification algorithm designed to \textit{incentivize} cheating, conceptually similar to the use of a negative threshold rule in Example \ref{negCut} to incentivize non-compliance. In Example \ref{twoCutExInner} the designer chooses an algorithm that is an \textit{inner} two-cut rule.

\ \\\noindent \textbf{Example \exampleNum{twoCutExInner}  \theexample .} \textit{Global Accuracy Improved by an Inner Two-Cut Rule.}\\~\\ 
Suppose that cheating ($\beta_i=\chi$) is possible, and the parameters of the setting are as follows:
\begin{eqnarray*}
r & = & 4,\\
c_i & \sim & \mathrm{Normal}(1/2,1/2),\\
\kappa & = & 0.6, \;\;\;\;\;\;\;\;\;\;\;\;\;\;\;\;\;\;\;\;\;\;\;\;\;\;\;\;\;\;\;\;\;\;\;\;\;\;\;\;\;\;\;\;\; \text{ and}\\
s_i \mid \beta_i & \sim & \begin{cases}
\mathrm{Normal}(0,1) & \text{ if }\beta = 0,\\
\mathrm{Normal}(1.5,1) & \text{ if }\beta = \chi,\\
\mathrm{Normal}(2,1) & \text{ if }\beta = 1.
\end{cases}
\end{eqnarray*}
In this setting, the sincere level of compliance is $F(0) = 16\%$ (as in Example \ref{negCut}).  Relative to the previous example with cheating (Example \ref{twoCutExOuter}), this population has comparatively high costs to compliance, and the optimal threshold rule is a negative threshold rule with $\tau^*=-1.5$. This rule will disincentivize compliance by committing to reward signals below $\tau^*$.  Of course, the downside of such a rule for accuracy is that it has to commit to misclassifying the tails of the signal distribution. Positively classifying an ``inner cut" mitigates some of these losses by positively classifying an inner slice of the signal distribution while still disincentivizing compliant behavior. We find that a more accurate inner two-cut rule assigns the positive classification outcome to individuals with signals on the interval $[0.85, 1.08]$. Figure \ref{Fig:TwoRulesWithCheatingInner} visualizes the most accurate (negative) threshold rule and better inner two-cut rule, along with the $\beta_i$-conditional PDFs.
\begin{figure}[h!]
\centering
\framebox{
\begin{tabular}{cc}
\includegraphics[scale=0.75]{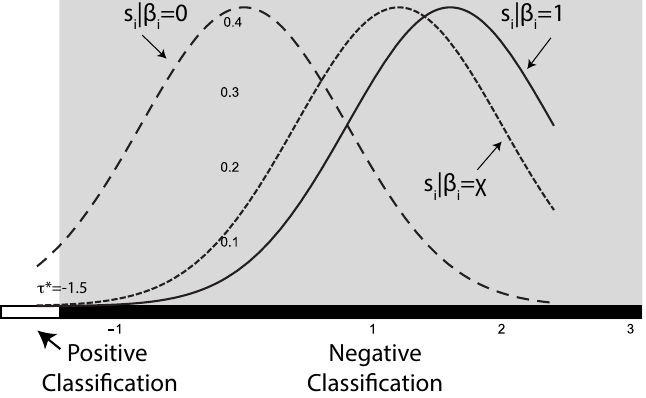} & \includegraphics[scale=0.75]{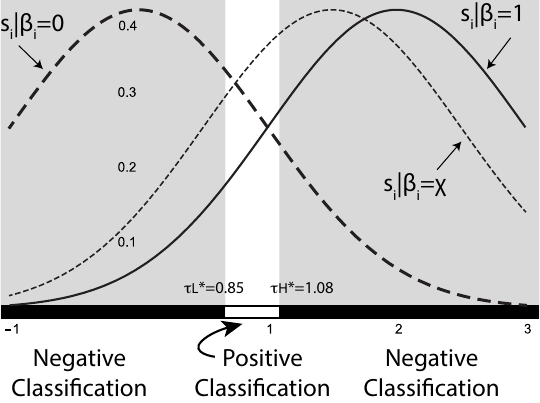}\\
\textbf{Optimal one-cut rule} & \textbf{Better two-cut rule}
\end{tabular}
}
\caption{Comparing Threshold and Inner Two-Cut Rules \label{Fig:TwoRulesWithCheatingInner}}
\end{figure}

Table \ref{Tab:InnerTwoGlobalAccuracy} reports the accuracy, compliance, and probability of cheating under these two rules.  It illustrates that the two classifiers produce similar levels of accuracy (with the two-cut rule representing a slight improvement) and virtually identical levels of compliance.  However, the two-cut rule \textit{induces} cheating behavior by about $18\%$ of the population that would have been non-compliant under a negative threshold rule. By committing to positively classify a closed interval of signals, this rule muddies its own distribution of data by inducing cheating behavior. However this is offset by improved classification on the left tail and on the mass of compliant types with signals on the interval. Again, at the margin the two-cut rule is more accurate in this setting. Table \ref{Tab:InnerTwoCutRuleGlobalAccuracyConfusionMatrix} displays the distributions over outcomes produced by the two rules, which are very similar (because cheating is equivalent to $\beta_i=0$ in the confusion matrix setting).
\begin{table}[H]
\centering
\begin{tabular}{|c|c||c|} \cline{2-3}
\multicolumn{1}{c|}{} & \textbf{Most Accurate} &\textbf{More Accurate} \\ 
\multicolumn{1}{c|}{} & \textbf{Negative Threshold} &\textbf{Inner Two-Cut Rule} \\ 
\multicolumn{1}{c|}{} & $\tau^*= -1.5$ & $(\tau^*_L, \tau^*_H)=(.85, 1.08)$\\ \hline
\textbf{Accuracy} &		87\%		&	88\%\\ \hline
\textbf{Compliance} &	6.5\%	&	6.5\% \\ \hline
\textbf{Cheating} &	$\approx$0\%	& 	18\%\\ \hline
\end{tabular}
\caption{Comparing Maximal Accuracy of Negative \& Inner Two-Cut Rules \label{Tab:InnerTwoGlobalAccuracy}}
\end{table}

\begin{table}[H]
\centering
\begin{tabular}{|c|c|c||c|c|} \cline{2-5}
\multicolumn{1}{c|}{}  & \multicolumn{2}{|c||}{\textbf{Negative Threshold Rule}} & \multicolumn{2}{c|}{\textbf{Inner Two-Cut Rule}}\\ \cline{2-5}
\multicolumn{1}{c|}{} &$d=1$&$d=0$& $d=1$&$d=0$\\ \hline
$\beta=1$& $\approx$0\% \,\,(TP) &7\%\,\,(FN)& $\approx$0\% \,\,(TP) &6\%\,\,(FN)\\  \hline
$\beta\neq 1$ &6\%\,\,(FP)& 87\%\,\,(TN)& 6\%\,\,(FP)& 88\%\,\,(TN)\\ \hline
\end{tabular}
\caption{Comparing Negative Threshold \& Inner Two-Cut Confusion Matrices \label{Tab:InnerTwoCutRuleGlobalAccuracyConfusionMatrix}}
\end{table}

\section{Conclusion \label{Sec:Conclusion}}

This article develops a theory of globally optimal classification when the behavior being classified responds strategically to the classifier.  Unlike standard learning settings in which the feature-label distribution is fixed, here the classifier induces incentives in a population, those incentives change behavior, and the ensuing behavior endogenously reshapes the distribution of signals on which the classifier operates.  We characterize the resulting Stackelberg equilibrium and show that, despite these feedback effects, globally optimal classifiers lie in a small and tractable family.

Across baseline environments, capacity constraints, and settings in which individuals can partially mimic compliance, optimal classifiers are either a single threshold rule or a two-cut rule.  These results show that strategic responses do not make optimal classification arbitrarily complex: the space of optimal classifiers remains low-dimensional, interpretable, and implementable.  The ``hard part'' of strategic classification is not the complexity of best responses but understanding how the classifier reshapes the population to which it is applied.

Recognizing that classification can reshape a population in potentially harmful ways leads us to also consider the behavioral effects of the classification objectives  guiding a designer. Different objectives induce different patterns of behavioral  feedback, and some classification objectives can make harmful equilibria inevitable.  We identify a class of designer objectives for which certain behavioral pathologies never arise in equilibrium: a designer with these objectives cannot rationally induce lower compliance in one group and higher compliance in another.  These results demonstrate that certain adjustments to classification objectives---such as adopting precision-seeking criteria rather than accuracy
maximizing criteria---can eliminate harmful equilibria.

Our analysis of cheating, or feature manipulation, also illustrates that optimal classification can deliberately shape the \textit{informational environment}.  A two-cut rule can make intermediate signals unattractive and induce agents to self-sort into more separated behaviors, yielding a cleaner distribution of data.  These non-monotone rules are unusual in machine learning and mechanism design, yet arise naturally in our framework. We note that similar incentive effects can be approximated by alternative, more familiar mechanisms---such as thresholds coupled with lotteries over scarce positive decisions---which can also deter cheating at the cost of some efficiency.

Taken together, our results highlight two design principles.  First, classification systems should be understood as incentive mechanisms and evaluated not only by their predictive performance but by the behavioral and distributional consequences they generate.  Second, fairness and efficiency tradeoffs may arise not from data imbalance or model misspecification but from the equilibrium feedback effects inherent in strategic environments.  More broadly, our framework offers a foundation for studying how different classification objectives can translate into behavioral, informational, and welfare outcomes in the populations to which classifiers are applied.

\newpage

\bibliographystyle{apsr}
\bibliography{john-book}

\appendix 

\section{Technical Appendix \label{Sec:Appendix}}

\subsection*{Notation Summary}

\begin{table}[H]
\centering
\begin{tabular}{ll}
\toprule
Symbol & Meaning \\
\midrule
$\beta$ & Underlying (unobserved) individual type \\
$s \in \mathcal{S}$ & Observed signal used for classification \\
$\delta(s) \in \{0,1\}$ & Classifier decision rule (reject/accept) \\
$c$ & Cost of compliance for an individual \\
$r$ & Net reward from receiving the positive classification \\
$F, f$ & CDF and PDF of compliance costs in the population \\
$G_0, G_1$ & Signal distributions conditional on behavior \\
$\mathrm{TPR}(\delta)$ & True positive rate induced by classifier $\delta$ \\
$\mathrm{FPR}(\delta)$ & False positive rate induced by classifier $\delta$ \\
$\Delta_\delta$ & Incentives gap $\mathrm{TPR}(\delta)-\mathrm{FPR}(\delta)$ \\
$\pi(\delta)$ & Compliance rate induced by $\delta$, or $F(r\cdot \Delta_\delta)$\\
\bottomrule
\end{tabular}
\caption{Baseline notation.}
\end{table}

\begin{table}[h!]
\centering
\begin{tabular}{ll}
\toprule
Setting & Optimal Classifier Structure \\
\midrule
Baseline (no quota, no cheating) & Single monotone threshold \\
Quota constraint & Single monotone threshold \\
Cheating feasible & Threshold rule or two--cut rule \\
\bottomrule
\end{tabular}
\caption{Summary of Optimal Algorithm Structure}
\end{table}

\subsection{\label{appProofs} Proofs}

For any given signal distribution $G=(G_0,G_1)$ we denote $\delta$'s true positive rate (TPR), false positive rate (FPR), true negative rate (TNR), and false negative rate (FNR) as follows:
\begin{eqnarray*}
TPR(\delta\mid G) & \equiv & \Pr[d_i=1\mid \beta_i=1, \delta, G]=\int_{\mathbf{R}} g_1(s) \cdot \delta(s) \dee s,\\
FPR(\delta\mid G) & \equiv & \Pr[d_i=1\mid \beta_i=0, \delta, G]=\int_{\mathbf{R}} g_0(s) \cdot \delta(s) \dee s,\\
TNR(\delta\mid G) & \equiv & \Pr[d_i=0\mid \beta_i=0, \delta, G]=\int_{\mathbf{R}} g_0(s) \cdot (1-\delta(s)) \dee s,\\
FNR(\delta\mid G) & \equiv & \Pr[d_i=0\mid \beta_i=1, \delta, G]=\int_{\mathbf{R}} g_1(s) \cdot (1-\delta(s)) \dee s.
\end{eqnarray*}
We restrict the classifiers to be measurable and almost everywhere continuous functions from $\mathbf{R}$ into the unit interval, $[0,1]$:
\[
\mathcal{D} \equiv \{\delta: \mathbf{R}\rightarrow[0, 1]\,\,\text{ a.e }\}\subset L^\infty(\mathbf{R}).
\]
For any signal distribution, $G$, we denote the set of $(TPR(\delta), FPR(\delta))$ pairs that can be achieved by an algorithm, $\delta \in \mathcal{D}$, as follows:
\[
\mathcal{P}_G=\{(TPR(\delta\mid G), FPR(\delta\mid G)): \delta \in\mathcal{D}\},
\]
and when no confusion can arise, we write $\mathcal{P} \equiv \mathcal{P}_G$. The following lemma enables us to conceptualize the designer's choice of an optimal algorithm $\delta(s)$ as a choice of any pair, $(TPR(\delta), FPR(\delta))$, subject to the signal distribution, $G=(G_0,G_1)$.  Lemma \ref{Slemma} allows us to treat the search over classifiers as a compact, convex optimization problem, ensuring that an optimal classifier exists.
\begin{lemma} 
\label{Slemma}
For any $G=(G_0,G_1)$ satisfying our assumptions, 
\begin{itemize}
    \item $\mathcal{D}$ is compact and convex in the weak-$\ast$ topology and 
    \item $\mathcal{P}$ is compact and convex in the usual Euclidean topology on $\mathbf{R}^2$. 
\end{itemize}
\end{lemma}
\begin{proof}
Our proof proceeds as follows.  We first prove convexity of $\mathcal{P}$ and $\mathcal{D}$ in the weak-$\ast$ topology and then prove that these spaces are compact in the weak-$\ast$ topology.\\~\\

\noindent \textit{Convexity of $\mathcal{P}$ and $\mathcal{D}$.} Take any $\delta^1, \delta^2\in\mathcal{D}$ and, for any $\lambda\in [0,1]$, define 
\[
\delta^\lambda \equiv \lambda \cdot \delta^1 + (1-\lambda) \cdot \delta^2.
\]
Because $\delta^1, \delta^2\in\mathcal{D}$ are each measurable and take on values in $[0, 1]$ almost everywhere, $\delta^\lambda$ also satisfies these properties. Because $\lambda \in [0,1]$ is arbitrary and $\delta^\lambda$ is a valid classifier for each $\lambda \in [0,1]$, this implies that $\mathcal{D}$ is convex. Similarly, 
\begin{eqnarray*}
TPR(\delta^\lambda) & = & \lambda \cdot TPR(\delta^1)+(1-\lambda) \cdot TPR(\delta^2), \;\;\;\;\; \text{ and }\\ 
FPR(\delta^\lambda) & = & \lambda \cdot FPR(\delta^1)+(1-\lambda) \cdot FPR(\delta^2),
\end{eqnarray*}
which, together, imply that $\left( TPR(\delta^\lambda), FPR(\delta^\lambda)\right) \in\mathcal{P}$.  Thus, $\mathcal{P}$ is also convex. \\

\noindent \textit{Compactness of $\mathcal{P}$ and $\mathcal{D}$.} Consider the normed space $L^1(\mathbf{R})$\footnote{$L^1(\mathbf{R})$ denotes the space of absolutely integrable real-valued functions on $\mathbf{R}$.} and its dual space, $L^\infty(\mathbf{R})$. Beginning with $\mathcal{D}$, note that $\mathcal{D} \subset L^\infty(\mathbf{R})$.  Endowing the space $L^\infty$ with the weak-$\ast$ topology implies the following characterization of convergence in $\mathcal{D}$:\footnote{The weak-$\ast$ topology on $\mathcal{D}$ identifies any $\delta$ with its equivalence class in $L^\infty(\mathbf{R})$, implying $0\leq\delta(s)\leq 1$ for almost all $s \in \mathbf{R}$.}
\[
\delta_n\rightarrow \delta \;\; \text{ if and only if } \;\; \int_\mathbf{R}f\cdot \delta_n \;\; \rightarrow \;\; \int_{\mathbf{R}}f\cdot\delta \text{ for every }f\in L^1.
\]
By the Banach-Alaoglu Theorem, the closed unit ball of a dual space is weak-$\ast$ compact. Since $L^\infty$ is the continuous dual of $L^1$ (\textit{i.e.}, $L^\infty=(L^1)^*$), the closed unit ball $\{\delta:||\delta||_\infty\leq 1\}$ is weak-$\ast$ compact. We know that $\mathcal{D}\subset \{\delta:||\delta||_\infty\leq 1\}$, and so our set of valid classifiers is a subset of a compact set.

We now show that $\mathcal{D}$ is closed in the weak-$\ast$ topology. Choose any convergent sequence $\delta_n$, with $\delta_j \in\mathcal{D}$ for all $j$, and let $\delta = \lim_{n\to \infty} \delta_n$ denote its limit.  We need to show that $\delta\in\mathcal{D}$, which will hold if $\delta(s) \in [0,1]$ for all $s \in \mathbf{R}$ except on a subset of $\mathcal{R}$ possessing Lebesgue measure zero. 

To do so, fix any nonnegative absolutely integrable function, $\phi \in L^1$. Because $\delta_n\in \mathcal{D}$ implies that $\delta_n(s)>0$ for almost all $s \in \mathbf{R}$, it must be the case that $\int_{\mathbf{R}}\phi(s) \cdot\delta_n(s) \dee s \geq 0$, which then implies that $\int_{\mathbf{R}}\phi(s) \cdot\delta(s) \dee s\geq 0$. 

Fix any set $E \subset \mathbf{R}$ possessing positive Lebesgue measure in $\mathcal{R}$, and any positive real number $Z>0$ such that $E\cap [-Z, Z]$ possesses positive Lebesgue measure in $\mathbf{R}$, and let \footnote{Formally, we intersect $E$ with $[-Z, Z]$ in order to ensure $\phi$ stays in $L^1$, because $\phi$ must integrate to a finite number.}
\[
\phi(s)=\begin{cases}
1 & \text{ if } s \in E\cap [-N, N], \\
0 & \text{ otherwise.}
\end{cases}
\]
It follows immediately that
\[
\int_{\mathbf{R}} \phi(s) \cdot \delta(s) \dee s<0,
\]
which results in a contradiction of the supposition that $\delta$ is the limit of a sequence $\delta_n$ with $\delta_n\in \mathcal{D}$ for all $n \in \{1,2,\ldots\}$.  Thus, we can conclude that $\delta(s)\geq 0$ for all $s$ except at most a subset of $\mathbf{R}$ possessing Lebesgue measure zero.  An identical argument proves $\delta\leq 1$.  Putting these together implies that, $\delta(s)\in [0,1]$ almost everywhere. Because $\delta_n$ is arbitrary, this establishes that $\mathcal{D}$ is closed in the weak-$\ast$ topology.  Finally, because $\mathcal{D}$ is a closed subset of a compact set, it thus follows that $\mathcal{D}$ is compact in the weak-$\ast$ topology.

Turning to $\mathcal{P}$, define
\[
T(\delta) \;\; \equiv \;\; (\; TPR(\delta), FPR(\delta)\; ) \;\; = \;\; \left(\int_{\mathbf{R}} g_1(s) \cdot \delta(s) \dee s, \int_{\mathbf{R}} g_0(s) \cdot \delta(s) \dee s\right).
\]
Note that the two coordinates of $T$ (and, hence, $T$ itself) are continuous in the weak-$\ast$ topology. It's immediate that $T$ is a weak-$\ast$ continuous mapping from $\mathcal{D}$ into $\mathbf{R}^2$. Putting these together, we've shown that  $\mathcal{P}=T(\mathcal{D})$ is the continuous image of a compact set. It follows that $\mathcal{P}$ is compact.
\end{proof}

\ \\ \textbf{Theorem \ref{quotaTheorem}}
\textit{Suppose $\delta$ can reward no more than $q\in(0, 1]$ of the individuals. The optimal classifier, $\delta_q^*$, is a threshold or negative threshold rule.}
\begin{proof}
Our proof proceeds in three steps. First, we discretize our signal space into a finite, ordered set of bins. Then, we set up and solve our constrained optimization problem on the discretized signal space, showing the optimal classifier is weakly monotonic (in the bin index, $k$) and involves non-degenerate randomization on at most one bin. To conclude, we then prove convergence (as the number of bins grows without bound) of the optimal binned-classifiers to the optimal classifier when the signal space is the continuum, $\mathbf{R}$.\\

 \noindent \textit{Step 1: Discretizing the Signal Space.} Fix $n \in \{1,2,\ldots\}$ and partition $\mathbf{R}$ into intervals, $\mathcal{B}_n = \{B^1, ..., B^n\}$, ordered as follows:
 \begin{eqnarray*}
 j<k & \Rightarrow & \sup[B^j] \leq \inf[B^j].
 \end{eqnarray*}
Then, for any $n \in \{1,2,\ldots\}$ and each bin $j \in \{1,\ldots,n\}$, define:\footnote{In words, $\gamma_{\beta, j}$ is the mass of the signal distribution $G_\beta$ on bin $B^j$, and $h^j$ is the difference between these $\beta$-conditional masses on $B^j$. Also note that, because $G_0$ and $G_1$ have full support on $\mathbf{R}$, $\min[\gamma_{0, j},\gamma_{1, j}]>0$ for all $j \in \{1,2,\ldots,n\}$.}
\begin{eqnarray*}
\gamma_{\beta, j} & \equiv & \int_{B^j} g_{\beta}(s) \; \dee s \;\;\;\;\; \text{and} \\
h^j & \equiv & \gamma_{1, j} - \gamma_{0, j}.
\end{eqnarray*}

We denote an arbitrary $\mathcal{B}_n$-coarsened rule by an $n$-dimensional vector, $\tilde{\delta}_n= (\tilde{\delta}_n^1, \tilde{\delta}_n^2,\ldots, \tilde{\delta}_n^n) \in [0,1]^n$, and we require that $\tilde{\delta}_n$ be measurable with respect to $\mathcal{B}_n$:
\[
s \in B^k \text{ and } s' \in B^k \Rightarrow \tilde{\delta}_n(s) = \tilde{\delta}_n(s').
\]
The true positive and false positive rates of $\tilde{\delta}_n$ are then defined as follows:
\begin{eqnarray*}
    TPR(\tilde{\delta}_n) & \equiv & \sum_{j=1}^n \gamma_{1, j} \cdot \tilde{\delta}_n^j, \text{ and}\\
    FPR(\tilde{\delta}_n) & \equiv & \sum_{j=1}^n \gamma_{0, j} \cdot \tilde{\delta}_n^j.
\end{eqnarray*}
As above, individual $i$'s incentive to choose $\beta_i=1$ under $\left\{\delta^j\right\}_{j=1}^n$ is defined as follows:
\[
\Delta_{\tilde{\delta}_n}\equiv \sum\limits_{j=1}^n h^j\cdot \tilde{\delta}_n^j = TPR(\tilde{\delta}_n) - FPR(\tilde{\delta}_n).
\]
To conserve space, we denote the prevalence induced by $\tilde{\delta}_n$ by
\[
\pi\left(\tilde{\delta}_n\right) = F\left( r\cdot \Delta_{\tilde{\delta}_n} \right).
\]

\noindent \textit{Step 2: Solving the Constrained Optimization Problem.} We now turn to the problem of finding the optimal classifier, $\tilde{\delta}_n^*$, given $\mathcal{B}_n$ and the quota, $q$. The Lagrangian for this constrained optimization problem is:
\begin{eqnarray*}
\mathcal{L}(\tilde{\delta}_n)&= &\underbrace{\sum_j \Big(\pi(\tilde{\delta}_n) \cdot \gamma_{1, j} \cdot (A_1\cdot \tilde{\delta}_n^j+A_0\cdot(1-\tilde{\delta}_n^j)) + 
(1-\pi(\tilde{\delta}_n)) \cdot \gamma_{0, j} \cdot (B_0\cdot\tilde{\delta}_n^j + B_1 \cdot (1-\tilde{\delta}_n^j))\Big)}_{\text{Designer's Objective}}\\
& & -\underbrace{\lambda \sum_j \Big( \pi(\tilde{\delta}_n) \cdot \gamma_{1,j} \cdot \tilde{\delta}_n^j + (1-\pi(\tilde{\delta}_n)) \cdot \gamma_{0, j} \cdot \tilde{\delta}_n^j -q \Big)}_{\text{Quota Constraint}} +\underbrace{\sum_j \mu^j \cdot (1-\tilde{\delta}_n^j)+\sum_j \nu^j \cdot \tilde{\delta}_n^j.}_{\text{Feasibility Constraint: } \tilde{\delta}_b^j \in [0,1]}
\end{eqnarray*}

\noindent The first order conditions with respect to $\tilde{\delta}_n^j$ are as follows:
\begin{equation}
\label{Eq:Thm1FOC}
\begin{aligned}
\frac{\partial \mathcal{L}}{\partial \tilde{\delta}_n^j}&= \pi(\tilde{\delta}_n) \cdot \gamma_{1, j}\cdot(A_1-A_0)-(1-\pi(\tilde{\delta}_n))\cdot \gamma_{0, j}\cdot(B_0-B_1)\\
&\;\;\;\;\;\;  -\lambda \cdot \Big( \pi(\tilde{\delta}_n) \cdot  \gamma_{1, j}+(1-\pi(\tilde{\delta}_n)) \cdot \gamma_{0, j} \Big)-\mu^j+\nu^j\\
&\;\;\;\;\;\; +r\cdot h^j\cdot 
f(r\cdot\Delta_{\tilde{\delta}_n}) \cdot \sum_k \Big(\gamma_{1, k} \cdot \big(A_1\cdot\tilde{\delta}_n^k+A_0\cdot(1-
\tilde{\delta}_n^k)\big)-\gamma_{0, k} \cdot \big(B_0\cdot\tilde{\delta}_n^k+B_1\cdot(1-\tilde{\delta}_n^k)\big)\Big)\\
&\;\;\;\;\;\; -\lambda\cdot r\cdot h^j\cdot f(r\cdot\Delta_{\tilde{\delta}_n})  \cdot \sum_k \Big( \gamma_{1, k}\cdot \tilde{\delta}_n^k- \gamma_{0, k}\cdot \tilde{\delta}_n^k\Big). 
\end{aligned}
\end{equation}
Defining
\begin{eqnarray*}
\xi & \equiv & r\cdot f(r\cdot\Delta_{\tilde{\delta}_n}) \Big(\big(A_1 \cdot TPR(\tilde{\delta}_n) + A_0 \cdot (1-FNR(\tilde{\delta}_n))\big)- \big(B_0 \cdot FPR(\tilde{\delta}_n) - B_1 \cdot (1-TNR(\tilde{\delta}_n))\big) \\
 & & -\lambda \cdot (TPR(\tilde{\delta}_n) - FPR(\tilde{\delta}_n))\Big),
\end{eqnarray*}
and, substituting $h^j=\gamma_{1, j}-\gamma_{0, j}$, Equation \ref{Eq:Thm1FOC} reduces to:
\begin{equation}
\label{Eq:PartialLagrangian}
{{\partial \mathcal{L}}\over{\partial \tilde{\delta}_n^j}}= \underbrace{\gamma_{1, j} \cdot \overbrace{\Big(\pi(\tilde{\delta}_n) (A_1-A_0-\lambda)+\xi \Big)}^{X} \; - \; \gamma_{0, j} \cdot \overbrace{\Big((1-\pi(\tilde{\delta}_n))(B_1-B_0+\lambda)  + \xi\Big)}^{Y}}_{W^j}-\mu^j+\nu^j,
\end{equation}
where, for simplicity (and as labeled in Equation \ref{Eq:PartialLagrangian}), we define
\begin{eqnarray*}
X & \equiv & \pi(\tilde{\delta}_n) \cdot (A_1-A_0-\lambda)+\xi,\\
Y & \equiv & (1-\pi(\tilde{\delta}_n)) \cdot (B_1-B_0+\lambda)  + \xi, \;\;\;\; \text{ and}\\
\text{ for each }  j \in \{1,\ldots,n\}, \;\;\;\; W^j & \equiv & \gamma_{1,j} \cdot X + \gamma_{0,j} \cdot Y.
\end{eqnarray*}
Note that, evaluated at any fixed vector $\tilde{\delta}_n$, the terms $X$ and $Y$ are constant across all bins $j$. \\~\\
The KKT necessary conditions for optimization imply each of the following about $\tilde{\delta}_n^*$:
\begin{enumerate}
\item $W^j-\mu^j+\nu^j=0$ for all $j$, %stationarity
\item $\lambda \cdot (v(TPR(\tilde{\delta}_n^*), FPR(\tilde{\delta}_n^*))\leq q$,
\item $\mu^j \cdot (1-\tilde{\delta}_n^{j*})=0$ for all $j$,
\item $\nu^j\cdot \tilde{\delta}_n^{j*}=0$ for all $j$, %complementary slackness
\item $\lambda\geq 0$ and $\mu^j\geq 0$ and $\nu^j \geq 0$ for all $j$. %dual feasibility
\end{enumerate}

\noindent Clearly $\tilde{\delta}_n^{j*} \in [0,1]$ for all $j$, and the KKT conditions yield the following conclusions: %Our delta\in[0, 1] by primal feasibility.
\begin{itemize}
\item If $\tilde{\delta}_n^{j*}=0$ then $\mu^j=0$, implying that $\nu^j=W^j\leq 0$.
\item If $\tilde{\delta}_n^{j*}=1$ then $\nu^j=0$, implying that $\mu^j=-W^j\leq 0$ (\textit{i.e.}, $W^j\geq 0$).
\item If $\tilde{\delta}_n^{j*}\in(0, 1)$ then $\mu^j=\nu^j=0$, implying that  $W^j=0$.
\end{itemize}

\noindent By the strict MLRP we know that ${{\gamma_{1, j}}\over{\gamma_{0, j}}}$ is strictly increasing in $j$. Because $X$ and $Y$ (as denoted in Equation \ref{Eq:PartialLagrangian}, above) are constant across bins at any fixed $\tilde{\delta}_n$, it follows that at an optimal $\tilde{\delta}_n^*$ the sign of $W_j$ can switch at most once in $j$. It follows that the optimal classifier is a threshold or negative threshold rule assigning $\tilde{\delta}_n^j=1$ when $W^j>0$, $\tilde{\delta}_n^j=0$ when $W^j<0$, and possibly randomizing on at most a single bin (the unique $j$ with $W^j=0$).\\~\\
\noindent \textit{Step 3: Enriching the Signal Space to $\mathbf{R}$.} Last, we construct a sequence of $n$-bin problems as above for increasing $n \to \infty$ and show that the optimal accuracy attained in these problems converges to the continuum optimum, and the optimal rules converge to a pure threshold rule that also meets the quota.\\~\\
Choose any sequence of nested bounded intervals, 
\[
\mathcal{I}_n\equiv (L_n,U_n] = \{ I_1, I_2, \ldots\}= \{ (L_1,U_1], (L_2, U_2], \ldots\},
\]
satisfying
\[
I_i \subset I_{i+1}, \;\;\;\; \forall i \in \{1,2,\ldots\}.
\]
For each $j \in \mathbb{N}$, define
\[
T_j \equiv \int_{\mathbf{R}\setminus \mathcal{I}_j}\big(g_0(s)+g_1(s)\big) \; \dee s,
\]
to represent the sum of the masses of $G_0$ and $G_1$ \textit{not} contained in $I_j$.  Each interval $I_j$ is partitioned into $b_j \in \mathbb{N}$ subintervals (``bins'') of equal Lebesgue measure (``width''), which we denote by $\epsilon_j > 0$.\\~\\
Now, without loss of generality, set $I_1 = (-1/2,1/2]$, $b_1 =1$, and  $\epsilon_1 = 1$.  Then let
\[
b \equiv \{b_1, b_2, \ldots\} = \{1, 3, 7, 15, \ldots\},
\]
or, more generally:
\[
b_{k+1} = 2 \cdot b_k + 1, \;\;\;\; \text{ for all } k \in \{2,3,\ldots\},
\]
denote the number of non-tail bins in $I_k$.  For each $k \in \mathbb{N}$, denote the set of bins for $I_k$ by
\[
\zeta_k \equiv \{\zeta_k^1, \ldots, \zeta_k^{b_k} \},
\]
satisfying the following:
\[
i<j \Rightarrow \sup\left[\,\zeta_k^i\,\right]\leq \inf \left[\,\zeta_k^j\,\right].
\]
Then, for each $k \in \{2,3,\ldots\}$, the width of the intervals in $I_k$, $\epsilon_k$, is defined as follows:
\[
\epsilon_{k+1} = \frac{\epsilon_k}{2},
\]
and define the interval $I_k \in \mathcal{I}_n$ as follows:
\[
I_{k+1} = (2L_k-U_k, 2U_k - L_k].
\]
For each $k \in \mathbb{N}$, denote the total signal mass outside interval $I_k$ by:
\[
T_k\equiv \int_{\mathbf{R}\setminus \mathcal{I}_k}\big(g_0(s)+g_1(s)\big) \; \dee s.
\]
Note the following:
\begin{enumerate}
    \item $s,s' \in \zeta_k^j \Rightarrow s,s' \in b^{j'}$ for some $j' \in \{1,\ldots, \zeta_{k'}\}$ for any $k' \in \{1,\ldots, k-1\}$,
    \item $\lim_{k \to \infty} \left[ I_k \right] = \mathbf{R}$,
    \item $\lim_{k \to \infty} \left[T_k\right]=0$, and
    \item $\lim_{k \to \infty} \epsilon_k = 0$.
\end{enumerate}
Now, for each $k \in \mathbb{R}$, let $\tilde{\delta}^{*}_{(k)}$ be the optimal classifier in the coarsened problem ($I_k$) constructed above, so $0 \leq \tilde{\delta}^{*}_{(k)} \leq 1$ everywhere and $v(\delta^{*}_{(k)}) = q$.  From the arguments in Step 2, above, we know that $\tilde{\delta}^{*}_{(k)}$ is a (positive or negative) threshold rule on the bin partition, using an interior (``mixing'') value for at most one bin.

We can view each $\tilde{\delta}^{*}_{(k)}$ as an element of $\mathcal{D}\subset L^\infty(\mathbb{R})$ (each $\tilde{\delta}^*_{(k)}$ is piecewise-constant on the bins and zero outside $I_k$). Because $\mathcal{D}$ is compact (Lemma \ref{Slemma}), the sequence $\{\tilde{\delta}^{*}_{(k)}\}_{k \in \mathbb{N}}$ contains a convergent subsequence, which we denote as $\{\tilde{\delta}^{*}_{(k_j)}\}_{j \in \mathbb{N}}$. Without loss of generality, we replace the original sequence with this convergent subsequence and re-index it, so we simply refer to it as $\{\tilde{\delta}^{*}_{(k)}\}_{k \in \mathbb{N}}$ for the rest of the proof. With this, we know that the limit of this sequence, 
\[
\tilde{\delta}^*_{\infty}  = \lim_{k \to \infty} \left[\tilde{\delta}^{*}_{(k)}\right] \in L^\infty,
\]
satisfies the following:
\[
\tilde{\delta}^{*}_{(k)} \rightarrow \tilde{\delta}^*_{\infty}\,\,\,\text{in the weak-}\ast\text{topology of} \,L^\infty.
\]
Note that each of the following functions:
\[
FPR(\delta)=\int_{\mathbf{R}} g_0(s) \cdot \delta(s) \; \dee s, \qquad 
TPR(\delta)=\int g_1(s) \cdot \delta(s) \; \dee s, \qquad
\Delta(\delta)=\int (g_1-g_0) \cdot \delta(s) \; \dee s,
\]
are weak-$\ast$ continuous because $g_0,g_1 \in L^1$.  Since $F(\cdot)$ is a continuous CDF, the quota mapping
\[
v(\delta) = (1-F(r \cdot \Delta(\delta)) \cdot FPR(\delta)
           +    F(r \cdot \Delta(\delta)) \cdot TPR(\delta)
\]
is also weak-$\ast$ continuous, as is the designer payoff $u_D(\delta)$, which only depends on those same three integrals.  Therefore:
\[
v(\tilde{\delta}^{*}_{(k)}) = q\,\,\Rightarrow\,\,
v(\tilde{\delta}^*_\infty)=q, \,\,\,\text{ and }\,\,\,
u_D\left(\tilde{\delta}^{*}_{(k)} \right) \rightarrow  u_D\left(\tilde{\delta}^*_\infty \right).
\]
It follows that $\tilde{\delta}^*_\infty$ is a feasible continuum rule that attains the limit of the optimal binned values.\\~\\
Now, let $V_n = u_D(\tilde{\delta}^{*}_{(n)})$ and 
\[
V^* = \sup \left\{u_D(\delta): v(\delta)=q,\; 0\le\delta\le 1 \right\}.
\]
Then, for any $k,k'\in \mathbb{N}$, with $k<k'$, we have
\[
I_k\subset I_{k'} \;\; \text{ and } \;\; \zeta_{k'}^j \in I_{k'}, \zeta_{k'}^j \in \zeta_k \cap \zeta_{k'}.
\]
In words, the intersection of the collection of non-tail bins for interval $I_{k'}$ are a refinement of the non-tail-bins in $I_k$ for any $k<k'$.  Accordingly, every rule that is admissible for interval $I_k$ is admissible for any $I_k'$ with $k<k'$, implying that $V_{n} \ge V_{n-1}$. 

Because every binned rule is a feasible continuum rule,  $V_n \le V^*$, so $V_n \uparrow V \le V^*$.  The previous paragraph shows that the limit value, $V$, is attained by $\tilde{\delta}^*_\infty$.  Accordingly, $V = V^*$ and $\tilde{\delta}^*_\infty$ is an optimal solution of the (limiting) continuum problem.

Finally, the arguments in Step~2 above imply that each $\tilde{\delta}^{*}_{(k)}$ is a threshold or negative--threshold rule with at most one mixing bin.   Because the bins in $I_{k'}\cap I_k$ are a partition of the non-tail support of $I_k$, the orientations $s_n\in\{\pm1\}$\footnote{These describe whether the threshold is positive ($s_n=1$) or negative ($s_n=-1$).} and thresholds, $\tau_n$, converge as $n \to \infty$.  Thus, the weak-$\ast$ limit, $\tilde{\delta}^*_\infty$, is a threshold rule on $\mathbf{R}$ that satisfies the quota constraint and attains the designer's optimum, as was to be shown.  
\end{proof}

\ \\  \textbf{ Theorem \ref{noNeg}} \textit{Let $[A_1\geq B_0]$, and $[A_0\geq B_1]$ and $[A_1\geq A_0$ or $B_1\geq B_0]$. Then an interior negative threshold rule is never optimal for the designer.} 
\begin{proof}
We will show that when designer preferences satisfy the above conditions, any interior negative threshold is strictly dominated by a rule assigning the positive classification to everyone or to no one. Begin by considering any interior negative threshold rule $\tau^- \in\mathbf{R}$. To simplify our expressions we define the following: 
$$
\begin{array}{lcl}
F^-&\equiv& F\big(r\cdot (G_1(\tau^-)-G_0(\tau^-))\big),\\
F_0&\equiv&F(0),\\
G^-_\beta&\equiv&G_\beta(\tau^-).
\end{array}
$$
In words, $F^-$ is the prevalence induced by negative threshold rule $\tau^-$, $F_0$ is sincere prevalence, and $G_\beta^-$ is the probability that an individual receives the positive classification outcome if choosing behavior $\beta$, given rule $\tau^-$. We know that $F^-<F_0$ because a negative threshold rule strictly reduces compliance relative to sincere compliance, and $G_0^->G_1^-$ by the MLRP.  We let $u^-$, $u_0$, and $u_1$ denote the designer's payoff at $\tau^-$ and at the classification rules assigning the negative outcome or positive outcome to everyone, respectively:

$$
\begin{array}{lcl}
u^-&\equiv&F^-\cdot \big(G_1^-\cdot A_1+(1-G_1^-)\cdot A_0 \big)+(1-F^-)\cdot \big(G_0^-\cdot B_0+(1-G_0^-)\cdot B_1 \big),\\
u_0&\equiv&F_0 \cdot A_0+(1-F_0)\cdot B_1,\\
u_1&\equiv&F_0 \cdot A_1+(1-F_0)\cdot B_0.
\end{array}
$$

We want to show that $u^-<\max\{u_1, u_0\}$. First, note that if $[A_1=B_0$ and $A_0=B_1]$ then the unique optimal rule assigns the positive classification to everyone if $A_1>A_0$ and to no one if $B_1>B_0$.\footnote{We have assumed that it can't be the case that $A_1=A_0=B_1=B_0$.} Consequently, we can restrict attention to cases in which one of these inequalities is strict: either $A_1>B_0$ or $A_0>B_1$ or both.

Next, define $\tilde{u}$ as the payoff the designer would receive if prevalence were fixed at  $F_0$ but $G_0^-$ and $G_1^-$ were unchanged (this payoff could likely not be obtained via any classification rule, but it provides us with a useful upper bound  on $u^-$):
$$\tilde{u}\equiv F_0\cdot \big(G_1^-\cdot A_1+(1-G_1^-)\cdot A_0 \big)+(1-F_0)\cdot \big(G_0^-\cdot B_0+(1-G_0^-)\cdot B_1 \big).$$ Our assumption that  $A_1\geq B_0$ and $A_0\geq B_1$ with one inequality strict  means that the designer strictly prefers greater compliance to less, conditional on holding $G_\beta^-$ constant. Consequently, $\tilde{u}> u^-$. This reduces our problem to showing that: $$\tilde{u}\leq \max\{u_0, u_1\}. $$ Suppose by way of contradiction that $\tilde{u}> \max\{u_0, u_1\}. $ This implies that the following two expressions are strictly positive:
\begin{equation}
\label{Eq:TildeUReexpression}
\begin{array}{lclcl}
\tilde{u}-u_1 &=&F_0 \cdot (1-G_1^-) \cdot (A_1-A_0)+(1-F_0) \cdot (1-G_0^-) \cdot (B_1-B_0)&>&0,\\
\tilde{u}-u_0 &=&F_0 \cdot G_1^- \cdot (A_1-A_0)+(1-F_0) \cdot G_0^- \cdot (B_0-B_1)&>&0.\\
\end{array}
\end{equation}
We will show that requiring the two above expressions to be positive forces $A_1<A_0$ and $B_1<B_0$, a contradiction. We'll begin by showing that $A_1<A_0$.

Taking a convex combination of the two inequalities in \eqref{Eq:TildeUReexpression} and rearranging yields
\begin{equation}
\label{firstEqPos} 
G_0^- \cdot (\tilde{u}-u_1)+(1-G_0^-)\cdot (\tilde{u}-u_0) = (A_0-A_1)\cdot (G_0^--G_1^-)\cdot F_0. 
\end{equation} 
We know that the left side of Equation \ref{firstEqPos} is strictly positive because it is a convex combination of two strictly positive terms. This convex combination reduces to the right side of Equation \ref{firstEqPos}, which requires $A_0>A_1$, as $F_0>0$ and $G_0^--G_1^->0$.

Now we show that $B_0>B_1$. Again taking a convex combination of the two inequalities in \eqref{Eq:TildeUReexpression} and rearranging yields:
\begin{equation}G_1^-(\tilde{u}-u_1)+(1-G_1^-)(\tilde{u}-u_0)=(B_0-B_1)(G_0^--G_1^-)(1-F_0). \label{secondEqPos} \end{equation}  Again, the left side of Equation \ref{secondEqPos} is strictly positive because it is a convex combination of two strictly positive terms, and it reduces to the right side of Equation \ref{secondEqPos}. This requires $B_0>B_1$, as $1-F_0>0$ and $G_0^--G_1^->0$.

It follows that if $\tilde{u}>\max\{u_0, u_1\}$ then it must be the case that $A_0>A_1$ and $B_0>B_1$, contradicting our assumption that either $A_1\geq A_0$ or $B_1\geq B_0$ (or both). It follows that $\tilde{u}\leq \max\{u_0, u_1\}$. Because $u^-< \tilde{u}$ we can conclude that either $u_1>  u^-$ or $u_0> u^-$, as was to be shown.
\end{proof}

\ \\  \textbf{ Proposition \ref{tightProp}} \textit{If any of the inequalities in Theorem \ref{noNeg} fails, then there exists an environment in which an interior negative threshold is optimal.}
\begin{proof}
We label our three inequalities as follows: 

\begin{center}(1) $A_1\geq B_0$, \,\,\,\,\,(2) $A_0\geq B_1$, \,\,\,\,\,and \,\,\,\,\,(3) [$A_1\geq A_0$ or $B_1\geq B_0$].\end{center}

For each inequality we numerically construct an environment (a choice of $A_1, A_0, B_1, B_0$, signal distributions $G_1$ and $G_0$, reward $r$, and cost distribution $F$) in which only that inequality is violated, and in which an interior negative threshold rule is strictly optimal for the designer. Throughout we will let:

$$\begin{array}{rcl}
F&=&\text{Normal}(0, 1),\\
G_0&=&\text{Normal}(0, 1),\\
G_1&=&\text{Normal}(1, 1).\\
\end{array}
$$

\ \\ \textit{Case 1: (1) and (2) are satisfied.} Let $(A_1, A_0, B_1, B_0)=\{\frac{1}{2}, 1, 0, \frac{1}{2}\}$ and let $r=\frac{1}{2}$. The optimal rule sets $\tau^{-*}=0.9$. This rule generates equilibrium compliance of $\pi^*=0.43$ and yields a designer payoff of $u_D(\tau^{-*})=0.56$. By comparison, the optimal positive threshold sets $\tau^* \approx \pm\infty$ (``punish all" or ``reward all") and yields the designer a payoff of $u_D(\tau^*)=0.5$. 

\ \\ \textit{Case 2: (2) and (3) are satisfied.} Let $(A_1, A_0, B_1, B_0)=\{\frac{1}{2}, 0, 0, 1\}$ and let $r=4$. The optimal rule sets $\tau^{-*}=1.52$. This rule generates equilibrium compliance of $\pi^*=.17$ and yields a designer payoff of $u_D(\tau^{-*})=0.84$. By comparison, the optimal positive threshold sets $\tau^*\approx -\infty$ (``reward all") and yields the designer a payoff of $u_D(\tau^*)=0.75$.

\ \\ \textit{Case 3: (1) and (3) are satisfied.} Let $(A_1, A_0, B_1, B_0)=\{0,\frac{1}{2}, 1, 0\}$ and let $r=6$. The optimal rule sets $\tau^{-*}=-1.3$. This rule generates equilibrium compliance of $\pi^*=0.3$ and yields a designer payoff of $u_D(\tau^{-*})=0.78$.  By comparison, the optimal positive threshold sets $\tau^*\approx \infty$ (``punish all") and yields the designer a payoff of $u_D(\tau^*)=0.75$. 
\end{proof}

\ \\  \textbf{ Theorem \ref{twoCuts}} \textit{Suppose that $G_1$, $G_\chi$, and $G_0$ share a common exponential form, as defined in Assumption \ref{twoCutAss}. The optimal classifier satisfying quota $q \in (0,1]$ is either a threshold rule or a two-cut rule.}
\begin{proof}
As in our proof of Theorem \ref{quotaTheorem}, we partition the signal line into ordered bins, $\mathcal{B}_n \equiv \{B^1, \ldots, B^n\}$, that are increasing in signal. For bin $B^j$ and behavior $\beta\in\{0, 1, \chi\}$, denote an arbitrary strategy in the resulting coarsened problem by $\tilde{\delta}_n = (\tilde{\delta}_n^1, \tilde{\delta}_n^2, \ldots, \tilde{\delta}_n^n)$, and define:
\[
\gamma_{\beta, j}\equiv \int_{B^j}g_{\beta}(s) \; \dee s \hspace{.25in}\text{and} \hspace{.25in} S_\beta \equiv \sum_j^n \, \gamma_{\beta, j} \cdot \tilde{\delta}_n^j.
\]
In words, $g_{\beta, j}$ is the mass of signal distribution $g_\beta$ on bin $B^j$, and $S_\beta$ is the probability a person choosing behavior $\beta$ receives a positive classification outcome, with $\tilde{\delta}^j$ representing the probability of a positive classification on bin $j$.\\~\\
To derive the optimal strategy for the coarsened problem, $\tilde{\delta}^*_n=\{\tilde{\delta}_n^{j*}\}_{j=1}^n$, we turn to the individual incentives induced by any coarsened strategy, $\tilde{\delta}_n$ as follows:
\begin{equation}
\label{Eq:BigDeltaCheatingProof}
\begin{array}{rcl}
\Delta^{1\chi}_{\tilde{\delta}_n} & \equiv & S_1(\tilde{\delta}_n) -S_\chi(\tilde{\delta}_n) ,\\[0.2em]
\Delta^{10}_{\tilde{\delta}_n} & \equiv & S_1(\tilde{\delta}_n) -S_0(\tilde{\delta}_n) , \\[0.2em]
\Delta^{\chi 0}_{\tilde{\delta}_n} & \equiv & S_\chi(\tilde{\delta}_n) -S_0(\tilde{\delta}_n) .
\end{array}
\end{equation}
The quantity $\Delta^{ab}_{\tilde{\delta}_n}$ equals the difference in the probability of $d_i=1$ when choosing behavior $\beta_i=a$ rather than $\beta_i=b$, given $\tilde{\delta}_n$. Accordingly $i$'s optimal behavior is: 
\begin{equation}
\label{behavior3Proof}
\beta_i^*(c_i \mid r,\tilde{\delta}_n,\kappa) \; = \; \begin{cases}
1 & \text{ if }\,\,\, c_i\leq \min\left[r\cdot \Delta^{10}_{\tilde{\delta}_n}, \;\; r\cdot \dfrac{\Delta^{1\chi}_{\tilde{\delta}_n}}{1-\kappa} \right],\\[1em]
\chi & \text{ if }\,\,\, c_i\in \left( r\cdot \dfrac{\Delta^{1\chi}_{\tilde{\delta}_n}}{1-\kappa} , \;\; r \cdot \dfrac{\Delta^{\chi 0}_{\tilde{\delta}_n}}{\kappa} \right),\,\,\, \\[1em]
0 & \text{ if }\,\,\, c_i\geq \max \left[r\cdot \Delta^{10}_{\tilde{\delta}_n}, \;\; r\cdot \dfrac{\Delta^{\chi 0}_{\tilde{\delta}_n}}{\kappa} \right], 
\end{cases}
\end{equation}
so that the prevalence of each behavior $\beta \in \{0,1, \chi\}$, is:
\[
\begin{array}{rcl}
\pi_1\left(\tilde{\delta}_n\right) & = & F\left(r \cdot \min\left[\Delta^{10}_{\tilde{\delta}_n}, \;\; \dfrac{\Delta^{1\chi}_{\tilde{\delta}_n}}{1-\kappa} \right] \right),\\[2em]
\pi_0\left(\tilde{\delta}_n\right) & = & 1-F\left(r \cdot 
\max \left[\Delta^{10}_{\tilde{\delta}_n}, \;\; \dfrac{\Delta^{\chi 0}_{\tilde{\delta}_n}}{\kappa} \right]\right), \,\,\,\text{ and }\\[2em]
\pi_\chi\left(\tilde{\delta}_n\right) & = & \begin{cases}
\; F\left( r\cdot \dfrac{\Delta^{1\chi}_{\tilde{\delta}_n}}{1-\kappa} \;\; r \cdot \dfrac{\Delta^{\chi 0}_{\tilde{\delta}_n}}{\kappa} \right), & \text{ if } \;\;\;
\dfrac{\Delta^{1\chi}_{\tilde{\delta}_n}}{1-\kappa} > \dfrac{\Delta^{\chi 0}_{\tilde{\delta}_n}}{\kappa},\\
\; 0 & \text{ otherwise.}
\end{cases}
\end{array}
\]

We can write our Lagrangian as follows:

\begin{align}
\mathcal{L}(\tilde{\delta}_n)=&\sum_j \Big(\pi_1\left(\tilde{\delta}_n\right) \cdot \gamma_{1,j} \cdot (A_1\cdot \tilde{\delta}_n^j+A_0\cdot(1-\tilde{\delta}_n^j))\Big) \label{cheatLagrangian}\\
&\hspace{1.25in} + \sum_j \Big( \pi_\chi\left(\tilde{\delta}_n\right) \cdot \gamma_{\chi,j}\cdot (B_0\cdot \tilde{\delta}_n^j + B_1 \cdot (1- \tilde{\delta}_n^j))\Big) \notag \\ 
 & \underbrace{\hspace{2.7in}+\sum_j\Big( \pi_0\left(\tilde{\delta}_n\right) \cdot \gamma_{0,j} \cdot (B_0\cdot\tilde{\delta}_n^j+B_1\cdot (1-\tilde{\delta}_n^j))\Big)\hspace{0in}}_{\text{Designer's objective}} \notag \\
&\hspace{.24in}-\underbrace{\lambda \cdot \Big(\sum_j \Big(\pi_1\left(\tilde{\delta}_n\right) \cdot \gamma_{1,j} \cdot \tilde{\delta}_n^j+ \pi_\chi\left(\tilde{\delta}_n\right) \cdot \gamma_{\chi,j}\cdot \tilde{\delta}_n^j+ \pi_0\left(\tilde{\delta}_n\right) \cdot \gamma_{0,j} \cdot \tilde{\delta}_n^j\Big) -q\Big)}_{\text{Quota constraint}} \notag \\
&\hspace{1.76in}+\underbrace{\sum_j \mu^j \cdot (1-\tilde{\delta}_n^j)+\sum_j \nu^j \cdot \tilde{\delta}_n^j.}_{\text{Constraints requiring } 0 \; \leq \; \tilde{\delta}_n^j \; \leq \; 1} \notag
\end{align}
Suppose that $\tilde{\delta}_n^*$ satisfies the KKT necessary conditions. If $\tilde{\delta}_n^*$ induces no cheating in equilibrium (\textit{i.e.}, $\Pr[\beta_i=\chi\mid \tilde{\delta}_n^*]=0$), then the problem is equivalent to that considered in the proof of Theorem 1, implying that $\tilde{\delta}_n^*$ must be a threshold rule. Accordingly, we now suppose that $\tilde{\delta}_n^*$ induces a positive equilibrium probability of cheating:
\[
\dfrac{\Delta^{1\chi}_{\tilde{\delta}_n}}{1-\kappa} > \dfrac{\Delta^{\chi 0}_{\tilde{\delta}_n}}{\kappa}.
\]
To conserve space, define the following terms:\footnote{Because $F$ is continuous but not necessarily smooth, we use directional derivatives of $F$ to calculate measurable selections $a_{\tilde{\delta}_n}\in\partial F(u_1)$ and $b_{\tilde{\delta}_n}\in\partial F(u_0)$. Directional derivatives exist everywhere and our KKT conditions hold with $a_{\tilde{\delta}_n}, b_{\tilde{\delta}_n}$ interpreted as any such selections.}
\begin{eqnarray*}
u^1_{\tilde{\delta}_n} & \equiv & \frac{r \cdot \Delta^{1\chi}_{\tilde{\delta}_n}}{1-\kappa},\\[0.3em]
u^0_{\tilde{\delta}_n} & \equiv & \frac{r \cdot \Delta^{\chi 0}_{\tilde{\delta}_n}}{\kappa},\\[0.3em]
a_{\tilde{\delta}_n} & \equiv & \frac{r \cdot f\left(u^1_{\tilde{\delta}_n}\right)}{1-\kappa} \;\;\;\; (\geq 0), \;\;\;\; \text{ and}\\[0.3em]
b_{\tilde{\delta}_n} & \equiv & \frac{r \cdot f\left(u^0_{\tilde{\delta}_n}\right)}{\kappa} \;\;\;\; (\geq 0).
\end{eqnarray*}
\ With this in hand, differentiating our Lagrangian yields:

\begin{equation*}
{{\partial \mathcal{L}}\over{\partial \tilde{\delta}_n^j}}= \underbrace{X\cdot \gamma_{1,j}+Y_\chi\cdot \gamma_{\chi,j}+Y_0\cdot \gamma_{0,j}}_{W^j}-\mu^j+\nu^j,
\end{equation*} 
where $X, Y_\chi$, and $Y_0$ are (bin-invariant) constants defined as follows:
\begin{eqnarray*}
X & \equiv & (A_1-A_0-\lambda)F(u^1_{\tilde{\delta}_n})+a_{\tilde{\delta}_n}\Big(S_1 \cdot (A_1-\lambda)+(1-S_1) \cdot A_0-S_\chi  \cdot (B_0-\lambda)+(1-S_\chi) \cdot B_1\Big),\\
Y_\chi & \equiv &(B_0-B_1-\lambda) \cdot (F(u^0_{\tilde{\delta}_n})-F(u^1_{\tilde{\delta}_n})) \\
&&-a_{\tilde{\delta}_n} \cdot \Big(S_1 \cdot (A_1-\lambda)+(1-S_1) \cdot A_0-S_\chi \cdot (B_0-\lambda)+(1-S_\chi) \cdot B_1\Big)\\
&&-b_{\tilde{\delta}_n} \cdot \Big(S_0 \cdot (B_0-\lambda)+(1-S_0) \cdot B_1-S_\chi \cdot (B_0-\lambda)+(1-S_\chi) \cdot B_1\Big), \;\;\;\;\;\;\;\;\;\; \\
Y_0 & \equiv & (B_0-B_1-\lambda) \cdot F(u^0_{\tilde{\delta}_n})+b_{\tilde{\delta}_n} \cdot \Big(S_0(B_0-\lambda)+(1-S_0) \cdot B_1-S_\chi  \cdot (B_0-\lambda)+(1-S_\chi) \cdot B_1\Big).
\end{eqnarray*}
As in our proof of Theorem 1, the KKT necessary conditions imply the following:
\begin{eqnarray}
\delta^j=0&\Rightarrow& W^j\leq 0, \nonumber \\
\delta^j=1&\Rightarrow& W^j\geq 0, \text{ and} \label{KKTcheat}\\
\delta^j\in(0, 1)&\Rightarrow & W^j= 0.\nonumber
\end{eqnarray}

Setting Equation \ref{KKTcheat} aside for the moment, let $\{\tilde{\delta}^*_{(n)}\}$ be our solutions for a sequence of binned problems. Theorem 1 implies that we can without loss of generality presume that this sequence is convergent with $\tilde{\delta}^*_\infty\in L^\infty$ and $0\leq \tilde{\delta}^*_\infty \leq 1$ almost everywhere such that $\tilde{\delta}^*_{(n)}\rightarrow \tilde{\delta}^*_\infty$. Also convergent in this sequence are (1) the proportion of people receiving the positive classification outcome and (2) the designer's expected payoff. Similarly, the fact that our bin-invariant terms, $(X_n, Y_{\chi, n}, Y_{0, n})$, are bounded implies that \[
(X_n, Y_{\chi, n}, Y_{0, n})\rightarrow (X, Y_\chi, Y_0).
\]  
Returning to Equation \ref{KKTcheat} we have in the continuum that:
\begin{equation}\label{linearEq}
W(s) \equiv X \cdot g_1(s)+Y_\chi \cdot g_\chi(s)+Y_0 \cdot g_0(s).
\end{equation}
Our KKT  conditions imply $\tilde{\delta}^*_\infty=1$ when $W>0$,  $\tilde{\delta}^*_\infty=0$ when $W<0$, and  $\tilde{\delta}^*_\infty \in[0, 1]$ when $W=0$.   Accordingly, the cut points of the limiting rule, $\tilde{\delta}^*_\infty$, correspond to zeros of $W$. \\~\\
Now, suppose that $b_0=\min\{b_1, b_\chi, b_0\}$.\footnote{This supposition is a normalization without loss of generality: one could use any of the three behaviors.}  Dividing $W(s)$ by $g_0>0$ preserves the zeroes of $W(s)$ and, by our assumption that the $g_\beta$ share a common exponential form,  yields $\tilde{W}(s)$:
\begin{eqnarray*}
\tilde{W}(s) & \equiv & X \cdot \dfrac{g_1(s)}{g_0(s)}+Y_\chi \cdot \dfrac{g_\chi(s)}{g_0(s)} +Y_0,\\[0.2em]
& = & X \cdot e^{a_1-a_0+(b_1-b_0)T(s)} +Y_\chi \cdot e^{a_\chi-a_0+(b_\chi-b_0)T(s)} +Y_0.    
\end{eqnarray*}
At this point, we can replace the limiting rule, $\tilde{\delta}_\infty^*$ with  $\delta^*$, that is equivalent in the weak-$\ast$ topology to $\tilde{\delta}_\infty^*$, which necessarily yields an equivalent $\tilde{W}(s)$. Then, taking the first order condition of $\tilde{W}$ yields:
\begin{equation}
\frac{d \tilde{W}}{d s}=e^{-a_0-b_0 \cdot T(s)}\cdot T^\prime(s) \cdot \left(X\cdot (b_1-b_0) \cdot e^{a_1+b_1\cdot T(s)}+ Y_\chi\cdot (b_\chi-b_0)e^{a_\chi+b_\chi \cdot T(s)}\right). \label{Eq:TildeWZeroesFOC} 
\end{equation}
Because $T$ is strictly monotonic,  $e^{-a_0-b_0 T(s)}\cdot T^\prime(s)$ is always either strictly positive or strictly negative. It follows that the slope of $\tilde{W}$ can change signs only when:
\begin{equation}
\label{stepZero}
X\cdot (b_1-b_0) \cdot e^{a_1+b_1\cdot T(s)}+ Y_\chi\cdot (b_\chi-b_0) \cdot e^{a_\chi+b_\chi\cdot T(s)}=0.
\end{equation} 
Again, dividing Equation \ref{stepZero} by $e^{a_\chi+b_\chi\cdot T(s)}>0$ retains the same zeroes as the original function. Then, letting
\[
b_{1 0} \equiv b_1-b_0, \;\;\;\;\;\;
b_{1 \chi} \equiv b_1 - b_\chi, \;\; \text{ and } \;\;\; b_{\chi 0} \equiv b_\chi - b_0,
\]
differentiating this new function yields:
\begin{equation}
\label{splittie}
\frac{d}{ds} \left[X \cdot b_{1 0}  e^{a_1-a_\chi+ b_{1 \chi} \cdot T(s)}+ Y_\chi  \cdot b_{\chi 0} \right] = 
X  \cdot T^\prime(s)  \cdot  b_{1 0} \cdot b_{1 \chi} e^{a_1-a_\chi+ b_{1 \chi} \cdot T(s)}.
%\frac{d}{ds} \left[X \cdot (b_1-b_0)  e^{a_1-a_\chi+(b_1-b_\chi) \cdot T(s)}+ Y_\chi  \cdot (b_\chi-b_0) \right] = 
%X  \cdot T^\prime(s)  \cdot  (b_1-b_0) \cdot (b_1-b_\chi) e^{a_1-a_\chi+(b_1-b_\chi) \cdot T(s)}.
\end{equation}

Since the right hand side of Equation \ref{splittie} has a constant sign, the derivative of $\tilde{W}(s)$ can change sign at most once. Consequently, $\tilde{W}(s)$ has at most one ``turning point" and can cross the $x-$axis at most twice.  As $\tilde{W}(s)$ was constructed to inherit the zeroes of our original function, $W(s)$, it follows that $W(s)$ can also have at most two zeros.

Finally, whether the optimal classifier, $\delta^*$, is a threshold rule or a two-cut rule depends on the values of the constants, $X$, $Y_0$, and $Y_\chi$ when evaluated at $\delta^*$, but $\delta^*$ must be either of the two, as was to be shown.

\end{proof}

\subsection{Discussion of Assumption \ref{twoCutAss}. \label{assDis}} The proof of Theorem~\ref{twoCuts} relies on the observation that the optimal
classifier can change its action only at points where a certain linear
combination of densities,
\[
W(s) \;=\; Y_0\, g_0(s) + X\, g_1(s) + Y_\chi\, g_\chi(s),
\]
changes sign (Equation \ref{linearEq}).  The constants $(X, Y_0, Y_\chi)$ depend on the parameters of the model.  Therefore, the number of times $W(s)$ crosses zero provides an upper bound on the number of cutpoints an optimal rule can have.  To guarantee that optimal classifiers require at most two cutpoints, we show that $W(s)$ has at most two zeros.  Assumption~\ref{twoCutAss} is a convenient sufficient condition for this
``two-crossing" property, but the property is not unique to exponential-family densities. Some log-concave location families also exhibit the same structure, as do certain ordered mixture distributions. We do not attempt to characterize the most general class of distributions for which the result holds. Our assumption is simply a tractable and empirically relevant environment that ensures that any optimal classifier remains low-dimensional. Environments where $W(s)$ exhibits more than two crossings are possible and may give rise to highly non-monotone optimal classifiers with multiple, disjoint decision regions.  These settings are substantively interesting but beyond the scope of this paper, as our focus is on conditions under which optimal mechanisms remain simple.

\end{document}